\documentclass[12pt]{amsart}  

\usepackage{graphicx}
\usepackage{here} 
\usepackage{array} 

\usepackage{vmargin}
\usepackage{amssymb}
\usepackage{mathrsfs}
\usepackage[all]{xy}
\usepackage[usenames,dvipsnames]{color}
\usepackage{soul}
\RequirePackage[colorlinks,linkcolor=blue,citecolor=LimeGreen,urlcolor=red]{hyperref} 
\usepackage{amsmath}

\newtheorem{theorem}{Theorem}[section]

\newtheorem{lemma}[theorem]{Lemma}

\newtheorem{prop}[theorem]{Proposition}
\newtheorem{remark}[theorem]{Remark}

\numberwithin{equation}{section}

\newcommand{\R}{\mathbb{R}}
\newcommand{\C}{\mathbb{C}}

\newcommand{\N}{\mathbb{N}}

\newcommand{\T}{\mathbb{T}}

\newcommand{\func}[3]{#1 : #2 \longrightarrow #3}

\newcommand{\abs}[1]{\left|#1\right|}
\newcommand{\eps}{\varepsilon}
\newcommand{\norm}[1]{\left\|#1\right\|}

\renewcommand{\leq}{\leqslant}
\renewcommand{\geq}{\geqslant}

\newcommand\restr[2]{{
  \left.\kern-\nulldelimiterspace 
  #1 
  \right|_{#2} 
  }}

\def\signmb{\bigskip \begin{center} {\sc
Marc Briant\par\vspace{3mm}
Brown University\par
Division of Applied Mathematics\par
182 George Street, Box F
Providence, RI 02192, USA\par
\vspace{3mm}
e-mail:} \tt{briant.maths@gmail.com} \end{center}}

\def\signcm{\bigskip \begin{center} {\sc
Cl\'ement Mouhot\par\vspace{3mm}
University of Cambridge\par
DPMMS, Centre for Mathematical Sciences\par
Wilberforce Road,
Cambridge CB3 0WA,
UK\par\vspace{3mm}
e-mail:} \tt{C.Mouhot@dpmms.cam.ac.uk} \end{center}}

\def\signsm{\bigskip \begin{center} {\sc
Sara Merino-Aceituno\par\vspace{3mm}
University of Cambridge\par
DPMMS, Centre for Mathematical Sciences\par
Wilberforce Road,
Cambridge CB3 0WA,
UK\par\vspace{3mm}
e-mail:} \tt{s.merino-aceituno@maths.cam.ac.uk} \end{center}}

\begin{document} 

\title[FROM BOLTZMANN TO INCOMPRESSIBLE NAVIER-STOKES IN SOBOLEV SPACES]{FROM BOLTZMANN TO INCOMPRESSIBLE NAVIER-STOKES IN SOBOLEV SPACES WITH POLYNOMIAL WEIGHT}
\author{M. Briant, S. Merino-Aceituno, C. Mouhot}

\maketitle

\begin{abstract}
We study the Boltzmann equation on the $d$-dimensional torus in a perturbative setting around a global equilibrium under the Navier-Stokes linearisation. We use a recent functional analysis breakthrough to prove that the linear part of the equation generates a $C^0$-semigroup with exponential decay in Lebesgue and Sobolev spaces with polynomial weight, independently on the Knudsen number. Finally we show a Cauchy theory and an exponential decay for the perturbed Boltzmann equation, uniformly in the Knudsen number, in Sobolev spaces with polynomial weight. The polynomial weight is almost optimal and furthermore, this result only requires derivatives in the space variable and allows to connect to solutions to the incompressible Navier-Stokes equations in these spaces.

\end{abstract}

\vspace*{10mm}

\textbf{Keywords:} Boltzmann equation on the Torus; Incompressible Navier-Stokes hydrodynamical limit; Cauchy theory uniform in Knudsen number, exponential rate of convergence towards global equilibrium. 


\smallskip
\textbf{Acknowledgements:}This work was supported by the UK Engineering and Physical Sciences Research Council (EPSRC) grant EP/H023348/1 for the University of Cambridge Centre for Doctoral Training, the Cambridge Centre for Analysis.
\tableofcontents

\section{Introduction} \label{sec:intro}

This article deals with the Boltzmann equation in a perturbative setting as the Knudsen number tends to zero. This equation rules the dynamics of rarefied gas particles moving on the flat torus in dimension $d$, $\T^d$, when the only interactions taken into account are binary collisions. More precisely, the Boltzmann equation describes the time evolution of the distribution $f=f(t,x,v)$ of particles in position $x$ and velocity $v$. A formal derivation of the Boltzmann equation from Newton's laws under the rarefied gas assumption can be found in \cite{Ce}, while \cite{Ce1} presents Lanford's Theorem (see \cite{La} and \cite{GST} for detailed proofs) which rigorously proves the derivation in short times.
\par We denote the Knudsen number by $\eps$ and the Boltzmann equation reads
$$\partial_tf + v\cdot\nabla_xf =\frac{1}{\eps} Q(f,f) \:,\: \mbox{on} \: \T^d \times \R^d,$$
where $Q$ is the Boltzmann collision operator given by
$$Q(f,f) =  \int_{\R^d\times \mathbb{S}^{d-1}}B\left(|v - v_*|,\mbox{cos}\:\theta\right)\left[f'f'_* - ff_*\right]dv_*d\sigma.$$
The Boltzmann kernel operator $B$ encodes the physics of the collision process and $f'$, $f_*$, $f'_*$ and $f$ are the values taken by $f$ at $v'$, $v_*$, $v'_*$ and $v$ respectively, where
$$\left\{ \begin{array}{rl}&\displaystyle{v' = \frac{v+v_*}{2} +  \frac{|v-v_*|}{2}\sigma} \vspace{2mm} \\ \vspace{2mm} &\displaystyle{v' _*= \frac{v+v_*}{2}  -  \frac{|v-v_*|}{2}\sigma} \end{array}\right., \: \mbox{and} \quad \mbox{cos}\:\theta = \left\langle \frac{v-v_*}{\abs{v-v_*}},\sigma\right\rangle .$$.

The Boltzmann collision operator comes from a symmetric bilinear operator $Q(g,h)$ defined by 
$$Q(g,h) =  \frac{1}{2}\int_{\R^d\times \mathbb{S}^{d-1}}B\left(|v - v_*|,\mbox{cos}\:\theta\right)\left[h'g'_* + h'_*g' - hg_*-h_*g\right]dv_*d\sigma.$$

\bigskip
It is well-known (see \cite{Ce}, \cite{Ce1} or \cite{Go} for example) that the global equilibria for the Boltzmann equation are the \textit{Maxwellians}, which are gaussian density functions depending only on the $v$ variable. Without loss of generality we consider only the case of normalized Maxwellians:
$$\mu(v) = \frac{1}{(2\pi)^{\frac{d}{2}}}e^{-\frac{\abs{v}^2}{2}}.$$

\bigskip
In this paper we will assume that the Boltzmann collision kernel is of the following form
\begin{equation}\label{kernelB}
B\left(|v - v_*|,\mbox{cos}\:\theta\right) = \Phi\left(|v - v_*|\right)b\left( \mbox{cos}\:\theta\right),
\end{equation}
with $\Phi$ and $b$ positive functions. This hypothesis is satisfied for all physical model and is more convenient to work with but do not impede the generality of our results.
\par We also restrict ourselves to the case of \textit{hard potential} or \textit{Maxwellian potential} ($\gamma=0$), that is to say there is a constant $C_\Phi >0$ such that
\begin{equation}\label{Phi}
\Phi(z) = C_\Phi z^\gamma, \quad \gamma \in [0,1],
\end{equation}
with a strong form of Grad's \textit{angular cutoff} (see \cite{Gr1}), expressed here by the fact that we assume $b$ to be $C^1$ with the controls from above
\begin{equation}\label{b}
\forall z \in [-1,1], \quad b(z),\: b(z') \leq C_b.
\end{equation}


\subsection{The problem and its motivations} \label{subsec:problem}
The Knudsen number is the inverse of the average number of collisions for each particle per unit of time. Therefore, as reviewed in \cite{Vi}, one can expect a convergence, in some sense, from the Boltzmann model towards the acoustics and the fluids dynamics as the Knudsen number tends to $0$. However, these different models describe physical phenomena that do not evolve at the same timescale and the right rescaling to approximate the incompressible Navier-Stokes equation  (see \cite{BGL}\cite{Go}\cite{Vi}\cite{Sa}) is the following equation

\begin{equation}\label{BE}
\partial_tf_\eps + \frac{1}{\eps}v \cdot \nabla_xf_\eps =\frac{1}{\eps^2} Q(f_\eps,f_\eps) \:,\: \mbox{on} \: \T^d \times \R^d,
\end{equation}
under the linearization $f_\eps(t,x,v) = \mu(v) + \eps  h_\eps(t,x,v)$. This leads to the perturbed Boltzmann equation
\begin{equation}\label{LBE}
\partial_th_\eps + \frac{1}{\eps}v\cdot\nabla_xh_\eps = \frac{1}{\eps^2}\mathcal{L}(h_\eps) + \frac{1}{\eps}Q(h_\eps,h_\eps),
\end{equation}
where we defined
$$\mathcal{L}(h) = 2Q(\mu,h).$$

\bigskip
The hydrodynamical limit of the perturbed equation is the system of equations satisfied by the limit, as $\eps$ tends to $0$, of the hydrodynamical fluctuations that are the following physical observables of $h_\eps$:
\begin{eqnarray*}
\rho_\eps(t,x) &=& \int_{\R^d}h_\eps(t,x,v)\:dv,
\\u_\eps(t,x)  &=& \int_{\R^d} vh_\eps(t,x,v)\:dv,
\\ \theta_\eps(t,x) &=& \frac{1}{d}\int_{\R^d} (\abs{v}^2-d)h_\eps(t,x,v)\:dv.
\end{eqnarray*}
Note that $(\rho_\eps, u_\eps,\theta_\eps)$ are the linearised fluctuations of the mass, momemtum and the thermal energy around the global equilibrium $\mu$.
\par In our perturbative framework, previous studies \cite{BGL}\cite{BU}\cite{Bri1} show that the hydrodynamical limits $\rho$, $u$ and $\theta$ are the weak (in the Leray sense \cite{Le}) solutions of the linearized incompressible Navier-Stokes equations:
\begin{eqnarray}
\partial_t u - \nu \Delta u + u\cdot \nabla u + \nabla p = 0, \nonumber
\\ \nabla \cdot u = 0, \label{NS}
\\ \partial_t \theta - \kappa \Delta \theta + u\cdot \nabla \theta = 0, \nonumber
\end{eqnarray}
where $p$ is the pressure function and $\nu$ and $\kappa$ are constants determined by $L$ (see \cite{BGL} or \cite{Go} Theorem $5$). They also satisfy the Boussineq relation

\begin{equation}\label{Boussineq}
\nabla(\rho + \theta) = 0.
\end{equation}

\bigskip
The aim of the present work is to use a constructive method to obtain existence and exponential decay for solutions to the perturbed Boltzmann equation $\eqref{BE}$, uniformly in the Knudsen number.  One will thus be allowed to extract a converging (at least weakly)  subsequence of $h_\eps$ converging to the incompressible Navier-Stokes equations \cite{BU}\cite{BGL}\cite{Bri1}. Such uniform results have been obtained on the torus in Sobolev spaces with exponential weight $H^s_{x,v}\left(\mu^{-1/2}\right)$ in \cite{Gu4}\cite{Bri1} and the present work improves this strong weight to a polynomial weight without the need of derivatives in the velocity variable.


\subsection{Existing results}\label{subsec:comparisonresults}

The first part of our work is to prove that the linear part of the Boltzmann equation
$$\mathcal{G}_\eps = \frac{1}{\eps^2}\mathcal{L} - \frac{1}{\eps} v\cdot\nabla_x$$
generates a strongly continuous semigroup with an exponential decay in Lebesgue and Sobolev spaces with polynomial weight, namely $1+\abs{v}^k$ for some $k$ large enough.
\par It has been known for long that the linear Boltzmann operator $\mathcal{L}$ is a self-adjoint non positive linear operator in the space $L^2_v\left(\mu^{-1/2}\right)$. Moreover it has a spectral gap $\lambda_0$. This has been proved in \cite{Ca2}\cite{Gr1}\cite{Gr2} with non constructive methods for hard potential with cutoff and in \cite{Bob1}\cite{Bob2} in the Maxwellian case. These results were made constructive in \cite{BM}\cite{Mo1} for more general collision operators. One can easily extend this spectral gap to Sobolev spaces $H^s_v\left(\mu^{-1/2}\right)$ (see for instance \cite{GMM} Section $4.1$). 

\bigskip
The next step is to see if the latter properties about $\mathcal{L}$ in the velocity space can be transposed when one adds the skew-symmetric transport operator $-v\cdot\nabla_x$. The first results were obtained in \cite{Uk} where $\mathcal{G}_1$ was proven to generate a strong continuous semigroup in $L^2_vH^s_x\left(\mu^{-1/2}\right)$ and in $L^\infty_vH^s_x\left(\mu^{-1/2}(1+\abs{v})^k\right)$, for $s$ and $k$ large enough. Then \cite{MN} obtained constructively this result in $H^s_{x,v}\left(\mu^{-1/2}\right)$ using hypocoercivity properties of the Boltzmann linear operator. Finally, a recent breakthrough proving abstract extension of semigroups \cite{GMM} showed that $\mathcal{G}_1$ generates a $C^0$-semigroup in all the Sobolev spaces of the form $W^{\alpha,q}_vW^{\beta,p}_x(m)$, for $m$ being an exponential weight (including maxwellian density if $q=p=2$) or a polynomial weight $(1+\abs{v})^k$, as long as $\alpha \leq \beta$ and $k$ is large enough depending on $q$ (with $k>2$ in the case $q=1$).

\bigskip
The full Boltzmann equation perturbed around a global equilibrium $\mu(v)$ $\eqref{LBE}$ has also been studied in the case $\eps=1$. The associated Cauchy problem has been worked on over the past fifty years, starting with Grad \cite{Gr}, and it has been studied in different spaces, such as  $L^2_vH^s_x\left(\mu^{-1/2}\right)$ spaces \cite{Uk} or $H^s_{x,v}\left(\mu^{-1/2}(1+\abs{v})^k\right)$ \cite{Gu1}\cite{Yu}. The Cauchy theory was then extended to $H^s_{x,v}\left(\mu^{-1/2}\right)$ where an exponential trend to equilibrium has also been obtained. This was obtained using hypocoercivity properties of the linear operator \cite{MN} or nonlinear estimates on fluid and microscopic parts of the equation \cite{Gu4}. Recently, \cite{GMM} proved existence and uniqueness for $\eqref{LBE}$ in more the general spaces $\left(W^{\alpha,1}_v \cap W^{\alpha,q}_v\right)W^{\beta,p}_x\left(1+\abs{v})^k\right)$ for $\alpha \leq \beta$ and $\beta$ and $k$ large enough with explicit thresholds. This result therefore gets rid of the exponential weight needed in the previous studies.
\par All the results presented above hold in the case of the torus. We refer the reader interested in the Cauchy problem, both for the torus and the whole space, to the review \cite{UkYa}.

\bigskip
For physical purposes, these studies for $\eps=1$ are relevant since mere rescalings or changes of physical units changes $\eqref{BE}$ to the case where the Knudsen number equals $1$. However, if one wants to study the hydrodynamical limits of the Boltzmann equation, one needs to obtain explicit dependencies on the Knudsen number. Using hypocoercivity methods \cite{Bri1} gave a constructive uniform approach on the semigroup generated by $\mathcal{G}_\eps$ in $H^s_{x,v}\left(\mu^{-1/2}\right)$ and its exponential decay. The study of the full perturbed Boltzmann equation $\eqref{LBE}$ taking into account the dependencies on the Knudsen number has been obtained \cite{Gu4}\cite{Bri1} in the same spaces $H^s_{x,v}\left(\mu^{-1/2}\right)$, for $s$ large enough. More precisely, for initial data sufficiently close to $\mu$ there exists a unique non-negative solution to $\eqref{BE}$ and it decays exponentially fast towards its equilibrium. The smallness assumption was proven to be independent of the Knudsen number as well as the rate of decay and the methods used in \cite{Bri1} are constructive.


\subsection{Our contributions and strategy}\label{subsec:strategy}

\bigskip
The present work brings two major improvements.
\par In the spirit of \cite{GMM}, we first prove that $\mathcal{G}_\eps$ generates a strong continuous semigroup in Sobolev spaces $W^{\alpha,1}_vW^{\beta,p}_x\left(1+\abs{v})^k\right)$ for $\alpha \leq \beta$ and $\beta$ and $k$ large enough with explicit thresholds. It is done by starting from existing results in $H^s_{x,v}\left(\mu^{-1/2}\right)$ and then decomposing the linear operator $\mathcal{G}_\eps$ into a dissipative part and a regularising part that is then treated in more and more regular spaces up to the space where the semigroup properties have been derived in previous articles. We thus improve the existing result \cite{Bri1}. Our main contribution is an adapted version of the abstract extension theorem developed in \cite{GMM} that takes into account the dependencies on the Knudsen number as well as a careful study of the dissipative and the  regularising parts of the operator $\mathcal{G}_\eps$.

\bigskip
The second contribution of this article is the solution to the Cauchy problem with exponential trend to equilibrium, independently on $\eps$, in spaces 
$$W^{\alpha,1}_vW^{\beta,1}_x\left(1+\abs{v}^{2+0}\right) \quad\mbox{and}\quad W^{\alpha,1}_vH^\beta_x\left(1+\abs{v}^{2+0}\right),$$
for $\beta$ large enough and all $\alpha \leq \beta$. First, this result makes the recent study \cite{GMM} uniform in the Knudsen number. Second, it improves the Cauchy theory developed uniformly in $\eps$ in \cite{Gu4}\cite{Bri1} by dropping the exponential weight and the $v$-derivatives. Moreover, one can notice that the polynomial weight is almost the optimal one for the Boltzmann equation (conservation of mass and energy). 
\par The main issue to obtain uniform results is that the bilinear operator $\eps^{-1}Q$ cannot be treated as a mere perturbation that evolves under the flow of $S_{\mathcal{G}_\eps}$, the semigroup generated by $\mathcal{G_\eps}$, since the latter has an exponential decay of order $O(1)$ that is negligeable compared to $O(\eps^{-1})$ as $\eps$ tends to zero. We develop an analytic point of view about the extension theorem in \cite{GMM} and include the bilinear term. More precisely, we decompose the perturbed equation $\eqref{LBE}$ into a hierarchy of equations taking place in spaces that have more and more regularity up to $H^s_{x,v}\left(\mu^{-1/2}\right)$ where estimates had been derived in \cite{Bri1}. At each step we use the dissipative part of the linear operator to control the remainder term $\eps^{-1}Q$ whereas the regularising part is controlled in spaces with higher regularity.


\subsection{Organization of the article}\label{subsec:organization}

Section $\ref{sec:mainresults}$ first introduces the different notations and definitions we will use throughout the paper and then states the precise theorems we prove in this work. Section $\ref{subsec:resultlinear}$ deals with the semigroup generated by the full linear operator $\eps^{-2}\mathcal{L}-\eps^{-1}v\cdot \nabla_x$ whereas Section $\ref{subsec:resultcauchy}$ is dedicated to the full Boltzmann equation.

\bigskip
The full linear part of the Boltzmann operator is proven to generate a strongly continuous semigroup in Lebesgue and Sobolev spaces with polynomial weight in Section $\ref{sec:lin_LqLp}$.
\par We start with Section $\ref{subsec:strategy_lin}$, a thorough description of our strategy and a version of the extension theorem of \cite{GMM} that takes into account the dependencies in $\eps$.
\par We show in this section that $\eps^{-2}\mathcal{L}-\eps^{-1}v\cdot \nabla_x$ can be decompose into a regularising operator in the velocity variable (Section $\ref{subsec:lin_decompositionL}$) and a dissipative one (Section $\ref{subsec:dissipativity}$).
\par We then combine the last two properties to gain regularity both in space and velocity (Section $\ref{subsec:iteratedconvolution}$) to finally prove the existence and exponential decay of the associated semigroup (Section $\ref{subsec:linproof}$).

\bigskip
The last section, Section $\ref{sec:cauchyexpodecay}$, proves existence, uniqueness and exponential decay of solutions to the perturbed Boltzmann equation $\eqref{LBE}$.
\par Section $\ref{subsec:descriptionpb}$ gives a new point of view on the extension we used to generate the semigroup associated to $\eps^{-2}\mathcal{L}-\eps^{-1}v\cdot \nabla_x$ and how it can be used with the bilinear operator. This strategy is developed through Sections $\ref{subsec:h0}$ and $\ref{subsec:h1}$ and it leads to the proof of the exponential decay towards equilibrium in Section $\ref{subsec:proofapriori}$. 

\section{Main results} \label{sec:mainresults}

\subsection{Notations}\label{subsec:notations}
We gather here the notations we will use throughout this article.

\bigskip
\textbf{Function spaces.} 
We first define the following shorthand notation,
$$\langle \cdot \rangle = \sqrt{1+\abs{\cdot}^2}.$$
\par The convention we choose is to index the space by the name of the concerned variable so we have, for $p$ in $[1,+\infty]$,
$$L^p_{[0,T]} = L^p \left([0,T]\right),\quad L^p_x = L^p\left(\T^d\right), \quad L^p_v = L^p\left(\R^d\right).$$
\par Let $p$ and $q$ be in $[1,+\infty)$, $\alpha$ and $\beta$ in $\N$  and $\func{m}{\R^d}{\R^+}$ a strictly positive measurable function. For any multi-indexes $j = (j_1,\dots,j_d)$ and $l= (l_1,\dots,l_d)$ in $\N^d$ we denote  the $(j,l)^{th}$ partial derivative by
$$\partial^j_l = \partial^l_x\partial^j_v.$$
 We define the space $W^{\alpha,q}_vW^{\beta,p}_x\left(m\right)$ by the norm
$$\norm{f}_{W^{\alpha,q}_vW^{\beta,p}_x\left(m\right)} = \sum\limits_{\overset{\abs{j}\leq\alpha,\abs{l}\leq \beta }{\abs{l}+\abs{j}\leq \max(\alpha,\beta)}} \norm{\left(\partial^j_l f\right) m}_{L^{q}_vL^{p}_x},$$
where we used the Lebesgue norm
$$\norm{g}_{L^{q}_vL^{p}_x} = \left[\int_{\R^d}\left(\int_{\T^d}\abs{f(x,v)}^p\:dx\right)^{q/p}\:dv\right]^{1/q}.$$

\bigskip
\textbf{Linear Boltzmann operator.} First we use a writing convention. The present work aims at extending results known in a small space $E$, namely $H^s_{x,v}\left(\mu^{-1/2}\right)$ with $s$ sufficiently large, into a larger space $\mathcal{E}$, namely Lebesgue and Sobolev spaces with polynomial weight. We will use curly letters for operators in $\mathcal{E}$ and their non-curly equivalent to denote their restriction to $E$. For instance, we will denote
$$\restr{\mathcal{L}}{E} = L.$$
\par The linear Boltzmann operator $L$ has several properties we will use throughout this paper (see \cite{Ce}\cite{Ce1}\cite{Vi2}\cite{GMM} for instance). 
\par $L$ is a closed self-adjoint operator in $L^{2}_v\left(\mu^{-1/2}\right)$ with kernel
$$\mbox{Ker}\left(L\right) = \mbox{Span}\left\{\phi_0(v),\dots,\phi_{d+1}(v)\right\}\mu ,$$
where $\phi_0(v)=1$, for $i=1,\dots,d$ we defined $\phi_i(v)=v_i$ and $\phi_{d+1} = \left(\abs{v}^2-d\right)/\sqrt{2d}$. The family $(\phi_i)_{0\leq i \leq d+1}$ is an orthonormal basis of $\mbox{Ker}\left(L\right)$ in $L^2_v\left(\mu^{-1/2}\right)$ and we denote $\pi_L$ the orthogonal projection onto $\mbox{Ker}\left(L\right)$ in $L^2_v\left(\mu^{-1/2}\right)$ :
\begin{equation}\label{piL}
\pi_L(h) = \sum\limits_{i=0}^{d+1} \left(\int_{\R^d} h(u)\phi_i(u)\:du\right)\phi_i(v)\mu(v),
\end{equation}
and we define $\pi_L^\bot = \mbox{Id} - \pi_L$. We will also denote the full linear Boltzmann operator by
$$\mathcal{G}_\eps = \frac{1}{\eps^2}L - \frac{1}{\eps}v\cdot \nabla_x.$$
For $s$ in $\N$ we will use the convention
$$\restr{\left(\mathcal{G}_\eps\right)}{H^s_{x,v}\left(\mu^{-1/2}\right)} = G_\eps.$$
It has been proven (\cite{Bri1} Proposition $3.1$) that the kernel of $G_\eps$ does not depend on $\eps$ and that its generators in $L^2_{x,v}\left(\mu^{-1/2}\right)$ are the same than the ones of $\mbox{Ker}\left(L\right)$. We therefore have that the orthogonal projection onto $\mbox{Ker}\left(G_\eps\right)$ in $L^2_{x,v}\left(\mu^{-1/2}\right)$ is given by
\begin{equation}\label{piGeps}
\Pi_{G}(h) = \Pi_{G_\eps}(h)= \sum\limits_{i=0}^{d+1} \left(\int_{\T^d\times\R^d} h(x,u)\phi_i(u)\:dxdu\right)\phi_i(v)\mu(v),
\end{equation}
and we define $\Pi_G^\bot = \mbox{Id} - \Pi_G$.
\par Note that for a function $h$ in $L^2_{x,v}\left(\mu^{-1/2}\right)$ we have that
$$\forall (x,v) \in \T^d\times\R^d, \quad \Pi_G(h) (x,v) = \int_{\T^d}\pi_L(h(x_*,\cdot))(v)\:dx_*.$$
\bigskip


\subsection{Results about the full linear part}\label{subsec:resultlinear}

We first deal with $\mathcal{G}_\eps$, the linear part of the perturbed Boltzmann operator. We prove that it generates a strongly continuous semigroup with an exponential decay in Lebesgue and Sobolev spaces with a weight $\langle v \rangle ^k$ as long as $k$ is large enough. The precise statement is the following.

\bigskip
\begin{theorem}\label{theo:linLqLp}
Let $B$ be a Boltzmann collision kernel satisfying $\eqref{kernelB}$-$\eqref{Phi}$-$\eqref{b}$.
\\There exists $0 < \eps_d\leq 1$ such that for all $p,\:q$ in $[1,+\infty]$, all $\alpha,\: \beta$ in $\N$ with $\alpha\leq \beta$ and all $k > k^*_q$, where
\begin{equation}\label{kq*}
k^*_q = \frac{3+\sqrt{49-48/q}}{2} +\gamma \left(1-\frac{1}{q}\right)
\end{equation}
with $\gamma$ defined in $\eqref{Phi}$,
\begin{enumerate}
\item for all $ 0 < \eps\leq \eps_d \:,\: \mathcal{G}_\eps = \eps^{-2}\mathcal{L} - \eps^{-1}v \cdot\nabla_x$ generates a $C^0$-semigroup, $S_{\mathcal{G}_\eps}(t)$, on $W^{\alpha,q}_vW^{\beta,p}_x\left(\langle v \rangle^k\right)$,
\item for all $\tau >0$, there exist $C_{\mathcal{G}}(\tau),\lambda_0 >0$, such that for all $0<  \eps \leq \eps_d$ and for all $h_{in}$ in $W^{\alpha,q}_vW^{\beta,p}_x\left(\langle v \rangle^k\right)$, for all $t\geq \tau$
 $$\norm{S_{\mathcal{G}_\eps}(t)(h_{in}) - \Pi_\mathcal{G}(h_{in})}_{W^{\alpha,q}_vW^{\beta,p}_x\left(\langle v \rangle^k\right)} \leq C_{\mathcal{G}}(\tau)e^{-\lambda_0 t} \norm{h_{in} - \Pi_\mathcal{G}(h_{in})}_{W^{\alpha,q}_vW^{\beta,p}_x\left(\langle v \rangle^k\right)},$$
\end{enumerate}
where $\Pi_\mathcal{G}$ is the spectral projector onto $\emph{\mbox{Ker}}\left(\mathcal{G}_\eps\right)$ which is given, for all $\eps$, by
\begin{equation}\label{projectionLqLp}
\Pi_\mathcal{G}(g) = \sum\limits_{i=0}^{d+1} \left(\int_{\T^d\times\R^d}g\phi_i \:dxdv\right)\phi_i \mu.
\end{equation}
The constants $\eps_d$, $C_\mathcal{G}(\tau)$ and $\lambda_0$ are constructive and only depends on $d$, $p$, $q$, $k$, $\alpha$, $\beta$ and the kernel of the Boltzmann operator.
\end{theorem}
\bigskip

We refer to \cite{Ka} and \cite{GMM} Section $2$ for definitions and properties of spectral projectors.

\bigskip
\begin{remark}\label{rem:linLqLp}
We can make a couple of remarks about this theorem.
\begin{enumerate}
\item It has been proven in \cite{Bri1} Section $3$, that in $H^1_{x,v}(\mu^{-1/2})$, $\emph{\mbox{Ker}}\left(G_\eps\right)$ does not depend on $\eps$ if $\eps$ is positive and we therefore can define $\Pi_G = \Pi_{G_\eps}$. During the proof of Theorem $\ref{theo:linLqLp}$ we will show that $\restr{\left(\Pi_{\mathcal{G}_\eps}\right)}{H^s_{x,v}\left(\mu^{-1/2}\right)} = \Pi_{G_\eps}$ and thus $\Pi_\mathcal{G}$ is well-defined and is independent of $\eps$ and given by $\eqref{piGeps}$.
\item As noticed in \cite{GMM}, the rate of decay $\lambda_0$ can be taken equal to the spectral gap of $\restr{\mathcal{L}}{H^s_{x,v}\left(\mu^{-1/2}\right)}$ (see \cite{Bri1}), for $s$ as large as wanted, when $k$ is big enough (and we obtained a constructive threshold).
\item Finally, we emphasize that in the case $q=1$, the result holds for all $k>2$. This is almost the minimal regularity $L^2_v\left(1+\abs{v}^2\right)$ for the Boltzmann equation, that is solutions with bounded mass and energy.
\end{enumerate}
\end{remark}
\bigskip


\subsection{Existence, uniqueness and trend to equilibrium}\label{subsec:resultcauchy}

A physically relevant requirement for solutions to the Boltzmann equation are that their mass, momentum and energy are preserved with time. This is also an \textit{a priori} property of the Boltzmann equation on the torus (see \cite{Vi2} Chapter $1$ Section $2$ for instance) which reads
$$\forall t \geq 0, \quad\int_{\T^d\times\R^d} \left(\begin{array} {c} 1 \\ v \\ \abs{v}^2 \end{array}\right) f_\eps(t,x,v) \:dxdv = \int_{\T^d\times\R^d} \left(\begin{array} {c} 1 \\ v \\ \abs{v}^2 \end{array}\right) f_\eps(0,x,v) \:dxdv.$$

\par If one expects trend to the equilibrium $\mu(v)$ for the solutions $f_\eps = \mu + \eps h_\eps$ of the Boltzmann equation $\eqref{BE}$ then it must be that
$$\forall t \geq 0, \quad \int_{\T^d\times\R^d} \left(\begin{array} {c} 1 \\ v \\ \abs{v}^2 \end{array}\right) h_\eps(t,x,v) \:dxdv = 0,$$
that is $\Pi_{\mathcal{G}_\eps}(h_\eps(t,\cdot,\cdot)) = 0$ for all $t$, which is a property that is indeed preserved along time for solution to the perturbed Boltzmann equation $\eqref{LBE}$, see \cite{Bri1} for instance.
\par We hence state the following theorem answering the Cauchy problem and the exponential convergence towards the equilibrium $\mu$.

\bigskip
\begin{theorem}\label{theo:cauchyexpodecay}
Let $B$ be a Boltzmann collision kernel satisfying $\eqref{kernelB}$-$\eqref{Phi}$-$\eqref{b}$ and let $p=1$ or $p=2$.

\noindent There exists $0<\eps_d\leq 1$ and $\beta_0$ in $\N$ such that 
\begin{itemize}
\item for all $\alpha$, $\beta$ in $\N$ such that $\beta\geq \beta_0$ and $\alpha \leq \beta$ and for all $k>2$ define
$$\mathcal{E}^{p} = W^{\alpha,1}_vW^{\beta,p}_x\left(\langle v \rangle^k\right),$$
\item for any $\lambda_0'$ in $(0,\lambda_0)$ ($\lambda_0$ defined in Theorem $\ref{theo:linLqLp}$) there exist $C_{\alpha,\beta},\:\eta_{\alpha,\beta}> 0$ such that for any $0<\eps \leq \eps_d$, for any distribution $0\leq f_{in} = \mu + \eps  h_{in}$:
\end{itemize}

\noindent If
\begin{enumerate}
\item[(i)] $h_{in}$ is in $\emph{\mbox{Ker}}(\mathcal{G}_\eps)^\bot$ in $\mathcal{E}^p$,
\item[(ii)] $\norm{h_{in}}_{\mathcal{E}^p} \leq \eta_{\alpha,\beta},$
\end{enumerate}
\bigskip
\noindent then there exists a unique global solution $f_\eps = f_\eps(t,x,v)$ to $\eqref{BE}$ in $\mathcal{E}^p$ which, moreover, satisfies $f_\eps = M + \eps h_\eps \geq 0$ with:
\begin{itemize}
\item $h_\eps$ belongs to $\emph{\mbox{Ker}}(\mathcal{G}_\eps)^\bot$ for all times,
\item $$\norm{h_\eps}_{\mathcal{E}^p} \leq C_{\alpha,\beta}\norm{h_{in}}_{\mathcal{E}^p} e^{-\lambda_0' t}.$$
\end{itemize}

\bigskip
\noindent The constants $C_{\alpha,\beta}$ and $\eta_{\alpha,\beta}$ are constructive and depends only on $\alpha$, $\beta$, $k$, $d$, $\lambda_0'$ and the kernel of the Boltzmann operator.
\end{theorem}
\bigskip

\section{The linear part: a $C^0$-semigroup in spaces with polynomial weight, proof of Theorem $\ref{theo:linLqLp}$}\label{sec:lin_LqLp}

In this section we focus on the linear part of the perturbed Boltzmann equation in $W^{\alpha,q}_vW^{\beta,p}_x\left(\langle v \rangle^k\right)$. We thus consider the following equation:

\begin{equation}\label{LBE1}
\partial_t h = \mathcal{G}_\eps (h).
\end{equation}


\subsection{Strategy of the proof}\label{subsec:strategylinLpLq}\label{subsec:strategy_lin}

If we denote $\mathcal{E} = W^{\alpha,q}_vW^{\beta,p}_x\left(\langle v \rangle^k\right)$ and $E = H^s_{x,v}\left(\mu^{-1/2}\right)$ we have that $E \subset \mathcal{E}$, dense with continuous embedding for $s$ large enough. \cite{Bri1} Theorem $2.1$ (with the norm of Theorem $2.4$) states that $G_\eps = \restr{\left(\mathcal{G}_\eps\right)}{E}$ generates a strongly continuous semigroup in $E$ with exponential decay. Theorem $\ref{theo:linLqLp}$ can therefore be understood as the possibility to extend properties of $G_\eps$ in a small space $E$ to  $\mathcal{G}_\eps$ in a larger space $\mathcal{E}$.
\par This issue of extending spectral gap properties as well as semigroup properties has been first tackled by Mouhot to obtain constructive rates of convergence to equilibrium for the homogeneous Boltzmann equation \cite{Mo3}. Recently, Gualdani, Mischler and Mouhot \cite{GMM} proposed a more abstract approach that allows to deal with the full linear operator. In their work, they proved that if some conditions on $\mathcal{G}_\eps$ and $G_\eps$ are satisfied then we can pass on some semigroup properties from $E$ to $\mathcal{E}$. The main argument of the proof of Theorem $\ref{theo:linLqLp}$ is to show that we can use their result in our setting, independently of $\eps$.
\par To be more precise, we give below a modified version of their main functional analysis theorem which is combination of Theorem $2.13$ and Lemma $2.17$ where we added dependencies on $\eps$.
\par We refer to \cite{GMM} Section $2$ for the definition of hypodissipativity (roughly speaking it is a dissipative property in a different norm on a Banach space) and the definition of the convolution of two semigroups of operators (denoted by the symbol $(*)$). In the sequel we will use $\mathscr{C}(E)$ for the set of closed operators on $E$ and $\mathscr{B}(E)$ for the set of bounded operators on $E$. For any operator $G$ in $\mathscr{C}\left(E\right)$ we denote $\mbox{R}(G)$ its range and $\Sigma (G)$ its spectrum.

\begin{theorem}[Modified extension theorem from \cite{GMM}]\label{theo:extension}
Let $\eps$ be a parameter such that $0<\eps\leq1$.
\\Let $E,\mathcal{E}$ be two Banach spaces with $E\subset\mathcal{E}$ dense with continuous embedding, and consider $G_\eps$ in $\mathscr{C}\left(E\right)$, $\mathcal{G}_\eps$ in $\mathscr{C}\left(\mathcal{E}\right)$ with $\restr{\left(\mathcal{G}_\eps\right)}{E} = G_\eps$ and $a>0$.
\par We assume the following
\begin{itemize}
\item[\emph{(A1)}] $G_\eps$ generates a semigroup $S_{G_\eps}$ on $E$, $G_\eps +a$ is hypodissipative on $\mbox{R}\left(\emph{\mbox{Id}} - \Pi_{G_\eps,a}\right)$ and 
$$\Sigma \left(G_\eps\right) \cap \{z\in \C,\: \emph{\mbox{Re}}(z)>-a\} = \{0\} \quad\mbox{is a discrete eigenvalue}.$$

\item[\emph{(A2)}] There exists $\mathcal{A}_\eps, \mathcal{B}_\eps$ in $\mathscr{C}\left(\mathcal{E}\right)$ such that $\mathcal{G}_\eps = \mathcal{A}_\eps + \mathcal{B}_\eps$ (with corresponding restrictions $A_\eps,B_\eps$ on $E$) and there exist some ``intermediate spaces'' (not necessarily ordered)
$$E = \mathcal{E}_J,\:\mathcal{E}_{J-1},\dots,\:\mathcal{E}_2,\:\mathcal{E}_1 = \mathcal{E}$$
such that, still denoting $\mathcal{B}_\eps:=\restr{\left(\mathcal{B}_\eps\right)}{\mathcal{E}_j}$ and $\mathcal{A}_\eps:=\restr{\left(\mathcal{A}_\eps\right)}{\mathcal{E}_j}$
   \begin{itemize}
   \item[(i)] $\left(\mathcal{B}_\eps + a/\eps^2 \right)$ is hypodissipative on $\mathcal{E}_j$;
   \item[(ii)] $\mathcal{A}_\eps \in \mathscr{B}\left(\mathcal{E}_j\right)$ with $\norm{\mathcal{A}_\eps}_{\mathscr{B}\left(\mathcal{E}_j\right)} \leq C_A/\eps^2$;
   \item[(iii)] there are some constants $l_0,l_1 \in \N^*$, $C\geq 1$, $K \in \R$, $\alpha \in [0,1)$ such that
   $$\forall t \geq 0, \quad \norm{T_{l_0}(t)}_{\mathscr{B}\left(\mathcal{E}_j,\mathcal{E}_{j+1}\right)} \leq C \frac{e^{Kt/\eps^2}}{\eps^{l_1}t^\alpha},$$
   for $1\leq j\leq J-1$, with the notation $T_l:=\left(\mathcal{A}_\eps S_{\mathcal{B}_\eps}\right)^{(*l)}$.
   \end{itemize}
\end{itemize}

Then $\mathcal{G}_\eps$ is hypodissipative in $\mathcal{E}$ and for all $a'<a$ there exists $n=n(a')\geq 1$ and some positive constants $C_{a'}$ and $C'_{a'}$ such that 

\begin{equation}\label{ineqTn}
\norm{T_{n}(t)}_{\mathscr{B}\left(\mathcal{E}\right)} \leq \frac{C_{a'}}{\eps^{nl_1/l_0}} e^{-a't/\eps^2};
\end{equation}
\begin{equation}\label{extensionsemigroup}
S_{\mathcal{G}_\eps}(t) = S_{G_\eps}(t)\Pi_{\mathcal{G}} + \sum\limits_{l=0}^{n-1}(-1)^l\left(\emph{\mbox{Id}}-\Pi_{\mathcal{G}} \right)S_{\mathcal{B}_\eps}*T_l(t) + (-1)^n\left[\left(\emph{\mbox{Id}}-\Pi_{\mathcal{G}} \right)S_{G_\eps}\right]*T_n(t);
\end{equation}
\begin{equation}\label{ineqextension}
\Big|\Big|S_{\mathcal{G}_\eps}(t) - S_{G_\eps}(t)\Pi_{\mathcal{G}}-(-1)^n\left[\left(\emph{\mbox{Id}}-\Pi_{\mathcal{G}} \right)S_{G_\eps}\right]*T_n(t)\Big|\Big|_{\mathscr{B}\left(\mathcal{E}\right)} \leq C'_{a'} \frac{t^n}{\eps^{n(2+l_1/l_0)}} e^{-a't/\eps^2},
\end{equation}
where $\Pi_{\mathcal{G}}$ has been defined in $\eqref{projectionLqLp}$.
\end{theorem}

We will use Theorem $\ref{theo:extension}$ to directly prove Theorem $\ref{theo:linLqLp}$. Indeed, \cite{Bri1} Theorem $2.1$ states that $G_\eps$ generates a strongly continuous semigroup with exponential decay in $E = H^s_{x,v}\left(\mu^{-1/2}\right)$, which is the required assumption $(A1)$ (properties about the spectral gap of the spectrum can be found in \cite{BM}). Therefore if $\mathcal{G}_\eps$ fulfils hypothesis $(A2)$ then it generates a strongly continuous semigroup, with an exponential decay of order $a'$ for all $a'<a$, since for all $\alpha,\:\beta,\:\eta >0$, all $t>t_0>0$ and all $0<2\eta'<\eta$,
\begin{equation}\label{ineqt0t}
\frac{t^\alpha}{\eps^{2\beta}}e^{-\eta\frac{t}{\eps^2}}\leq C_{\beta,\eta'} t^{\alpha-\beta}e^{-(\eta-\eta')\frac{t}{2\eps^2}}\leq C_{\beta,\eta'} t^{\alpha-\beta}e^{-(\eta-\eta')t} \leq C_{t_0,\alpha,\beta,\eta'}e^{-(\eta-2\eta')t},
\end{equation}
for $0<\eps\leq 1$.


\subsection{Decomposition of the operator and assumption $(A2)(ii)$}\label{subsec:lin_decompositionL}

In this section we find a decomposition $\mathcal{G}_\eps = \mathcal{A}_\eps +\mathcal{B}_\eps$ that will fit the requirements $(A1)-(A2)$ of Theorem $\ref{theo:extension}$. This decomposition has been found in \cite{GMM} in the case $\eps = 1$. We will use exactly the same operators but including the dependencies in $\eps$. All the results presented in the rest of this section are true for $\eps=1$ (see \cite{GMM} Section $4$) so we will try to relate as much as possible our computations with the ones for $\eps=1$.

\bigskip
For $\delta$ in $(0,1)$, to be chosen later, we consider $\Theta_\delta = \Theta_\delta(v,v_*,\sigma)$ in $C^\infty$ that is bounded by one on the set
$$\left\{\abs{v}\leq \delta^{-1}    \quad\mbox{and}\quad 2\delta\leq\abs{v-v_*}\leq \delta^{-1}    \quad\mbox{and}\quad \abs{\mbox{cos}\:\theta} \leq 1-2\delta \right\}$$
and whose support is included in
$$\left\{\abs{v}\leq 2\delta^{-1}    \quad\mbox{and}\quad \delta\leq\abs{v-v_*}\leq 2\delta^{-1}    \quad\mbox{and}\quad \abs{\mbox{cos}\:\theta} \leq 1-\delta \right\}.$$
We define the splitting
$$\mathcal{G}_\eps = \mathcal{A}^{(\delta)}_\eps +\mathcal{B}^{(\delta)}_\eps,$$
with
$$\mathcal{A}^{(\delta)}_\eps h (v) = \frac{1}{\eps^2}\int_{\R^d\times\mathbb{S}^{d-1}}\Theta_\delta\left[\mu'_*h' + \mu'h'_* - \mu h_*\right]b\left(\mbox{cos}\:\theta\right)\abs{v-v_*}^\gamma\:d\sigma dv_*$$
and
$$\mathcal{B}^{(\delta)}_\eps h (v) = \mathcal{B}^{(\delta)}_{2,\eps} h (v) -\frac{1}{\eps^2}\nu(v) h(v) -\frac{1}{\eps}v\cdot\nabla_x h(v),$$
where
$$\mathcal{B}^{(\delta)}_{2,\eps} h (v) = \frac{1}{\eps^2}\int_{\R^d\times\mathbb{S}^{d-1}}\left(1-\Theta_\delta\right)\left[\mu'_*h' + \mu'h'_* - \mu h_*\right]b\left(\mbox{cos}\:\theta\right)\abs{v-v_*}^\gamma\:d\sigma dv_*$$
and $\nu(v)$ is the standard collision frequency
$$\nu(v) = \int_{\R^d\times\mathbb{S}^{d-1}} b\left(\mbox{cos}\:\theta\right)\abs{v-v_*}^\gamma \mu_*\:d\sigma dv_*.$$
Note that there exists $\nu_0,\:\nu_1 >0$ such that
\begin{equation}\label{nu0nu1}
\forall v \in \R^d,\quad \nu_0(1+\abs{v}^\gamma)\leq \nu(v)\leq \nu_1(1+\abs{v}^\gamma).
\end{equation}

\bigskip
We have that 
$$\mathcal{A}^{(\delta)}_\eps = \frac{1}{\eps^2}\mathcal{A}^{(\delta)}_1 \quad\mbox{and}\quad \mathcal{B}^{(\delta)}_{2,\eps} = \frac{1}{\eps^2}\mathcal{B}^{(\delta)}_{2,1}.$$
We therefore obtain the following controls on $\mathcal{A}^{(\delta)}_\eps$.

\bigskip
\begin{prop}\label{prop:controlAeps}
For all $0<\eps<\eps^d$, for any $q$ in $[1,+\infty]$ and $\alpha\geq 0$, the operator $\mathcal{A}^{(\delta)}_\eps$ maps $L^q_v$ into $W^{\alpha,q}_v$ with compact support.
\\ There exists $C_{\delta,\alpha,q}, R_{\delta} >0$ independent of $\eps$ such that 
$$\forall h \in L^q_v,\: \emph{\mbox{supp}}\left(\mathcal{A}^{(\delta)}_\eps h \right) \subset B(0,R_{\delta}), \quad \norm{\mathcal{A}^{(\delta)}_\eps h}_{W^{\alpha,q}_v} \leq \frac{C_{\delta,\alpha,q}}{\eps^2}\norm{h}_{L^q_v}.$$
Moreover, for any $p$ in $[1,+\infty]$ and for all $h$ in $L^q_vL^p_x$,
$$\norm{\mathcal{A}^{(\delta)}_\eps h}_{L^q_vL^p_x} \leq \norm{\mathcal{A}^{(\delta)}_\eps\left(\norm{h}_{L^p_x}\right)}_{L^q_v}$$
\end{prop}
\bigskip

\begin{remark}\label{rem:Adeltasobolev}
We notice here that this Proposition gives the point $(A2)(ii)$ of Theorem $\ref{theo:extension}$ if the $\mathcal{E}_j$ are Sobolev spaces.
\end{remark}
\bigskip

\begin{proof}[Proof of Proposition $\ref{prop:controlAeps}$]
The kernel of the operator $\mathcal{A}^{(\delta)}_\eps$ is of compact support so its Carleman representation (see \cite{Ca2}) gives the existence of $k^{(\delta)}$ in $C^\infty_c\left(\R^d\times\R^d\right)$ such that
\begin{equation}\label{Adeltaepskernel} 
 \mathcal{A}^{(\delta)}_\eps h(v) = \frac{1}{\eps^2}\int_{\R^d} k^{(\delta)}(v,v_*)h(v_*)\:dv_*,
 \end{equation}
and therefore the control on $\norm{\mathcal{A}^{(\delta)}_\eps h}_{W^{\alpha,q}_v}$ is straightforward.

\bigskip
The control of $\norm{\mathcal{A}^{(\delta)}_\eps h}_{L^q_vL^p_x}$ comes directly from Minkowski's integral inequality which states
$$\left[\int_{\T^d}\left(\int_{\R^d}k^{(\delta)}(v,v_*)h(x,v_*)dv_*\right)^p dx\right]^{1/p} \leq \int_{\R^d}\left(\int_{\T^d}k^{(\delta)}(v,v_*)^p h(x,v_*)^p dx\right)^{1/p}dv_*.$$
\end{proof}


\subsection{Dissipativity estimates for $\mathcal{B}^{(\delta)}_\eps$, assumption $(A2)(i)$}\label{subsec:dissipativity}

One can find in \cite{GMM} proof of Lemma $4.14$ case $(W2)$ and $(W3)$ the following estimate on the operator $\mathcal{B}^{(\delta)}_\eps$ in the case $\eps=1$.

\bigskip
\begin{lemma}\label{lem:controlB1}
For all $p,q$ in $[1,+\infty]$, for all $k>2$ and for any $\delta$ in $(0,1)$ and all $h$ in $L^q_vL^p_x\left(\langle v \rangle^k\right)$,

$$\int_{\R^d}\langle v\rangle^{kq} \norm{h}^{q-p}_{L^p_x}\left(\int_{\T^d}\mbox{sgn}(h)\abs{h}^{p-1}\mathcal{B}^{(\delta)}_1 h \:dx\right)\:dv \leq \left[\Lambda_{k-\gamma/q',q}(\delta)-1\right]\norm{h}^q_{L^q_vL^p_x\left(\langle v \rangle^k \nu^{1/q}\right)},$$
where $q'$ is the conjugate exponent of $1/q$ and $\Lambda_{k,q}(\delta)$ is a constructive constant such that
$$\lim\limits_{\delta\to 0} \Lambda_{k,q}(\delta) = \phi_q(k) =\left(\frac{4}{k+2}\right)^{1/q}\left(\frac{4}{k-1}\right)^{1-1/q}.$$
\end{lemma}
\bigskip

\begin{remark}\label{rem:lambdakq}
As noticed in \cite{GMM} Remark $4.3$, the quantity $\phi_q(k)$ is strictly less than one for $k$ bigger than a constant $k^{**}_q$. The constant $k^*_q$ we are considering is not optimal  and is such that $\phi_q(k-\gamma/q') <1$, where $q'$ is the conjugate exponent of $q$. This appearance of $k-\gamma/q'$ is due to a loss of weight of order $\nu^{-1/q'}$ in the estimate of the spectral gap, see proof of Proposition $\ref{prop:hypodissipativityBeps}$.
\end{remark}
\bigskip

In the case of the Boltzmann operator with hard potential and angular cutoff, point $(A2)(i)$ is fulfilled by $\mathcal{B}^{(\delta)}_{\eps}$ for $\delta$ small enough. This is the purpose of the following lemma. We recall here that $\nu_0 = \inf\limits_{v \in \R^d}(\nu(v)) >0$ and that we define
$$\norm{\cdot}_{W^{\alpha,q}_vW^{\beta,p}_x\left(\langle v \rangle^k\right)} = \sum\limits_{\overset{ \abs{l}+\abs{j}\leq \max(\alpha,\beta)}{\abs{j}\leq\alpha,\abs{l}\leq \beta}} \norm{\partial^j_l \cdot}_{L^{q}_vL^{p}_x\left(\langle v \rangle^k\right)}.$$

\bigskip
\begin{prop}\label{prop:hypodissipativityBeps}
Consider $p,q$ in $[1,+\infty]$,  $k>k_q^*$, defined by $\eqref{kq*}$, and $\alpha$, $\beta$ in $\N$ such that $\alpha \leq \beta$.
\\Then there exists $\delta_{k,q}$ in $(0,1)$ such that for all $0<\delta \leq \delta_{k,q}$ there exists $\lambda_0=\lambda_0(k,q,\delta)$ in $(0,\nu_0)$ such that for all $0<\eps \leq 1$,
\begin{itemize}
\item $\lambda_0(k,q,\delta)$ tends to $\lambda^*_0(k,q)$ as $\delta$ goes to $0$,
\item $\lambda^*_0(k,q)$ tends to $\nu_0$ when $k$ goes to $+\infty$,
\item $\left(\mathcal{B}^{(\delta)}_{\eps} + \lambda_0/\eps^2\right)$ is dissipative in $W^{\alpha,q}_vW^{\beta,p}_x\left(\langle v \rangle^k\right)$.
\end{itemize}
\end{prop}
\bigskip

\begin{proof}[Proof of Proposition $\ref{prop:hypodissipativityBeps}$]
Let $h_0$ be in $W^{\alpha,q}_vW^{\beta,p}_x\left(\langle v \rangle^k\right)$ and considert $h$ to be a solution to the linear equation
\begin{equation}\label{EDPdissipativity}
\partial_t h  = \mathcal{B}^{(\delta)}_\eps h = \mathcal{B}^{(\delta)}_{2,\eps} h -\frac{1}{\eps^2}\nu h -\frac{1}{\eps}v\cdot\nabla_x h,
\end{equation}
with initial value $h_0$.
\par Since the $x$-derivative commutes with the equation we can consider only the case when $\beta=\alpha$. The proof is split into two parts. First we prove Proposition $\ref{prop:hypodissipativityBeps}$ in the case $\alpha=0$ and then we study the case with $v$-derivatives.

\bigskip
\textbf{Step $1$: the case $\alpha = 0$}. Take $p,q$ in $[1,+\infty)$. 
\par We recall that
$$\norm{h}_{L^q_vL^p_x\left(\langle v \rangle ^k\right)} = \left[ \int_{\R^d}\left(1+\abs{v}^k\right)^q \left(\int_{\T^d}\abs{h}^p \:dx\right)^{q/p}\:dv\right]^{1/q}.$$
Therefore we can compute

\begin{equation}\label{dissipativitystart}
\begin{split}
\frac{d}{dt}\norm{h}_{L^q_vL^p_x\left(\langle v \rangle ^k\right)} =& \norm{h}_{L^q_vL^p_x\left(\langle v \rangle ^k\right)}^{1-q} 
\\&\times\int_{\R^d}\left(1+\abs{v}^k\right)^q \norm{h}^{q-p}_{L^p_x}\left(\int_{\T^d}\mbox{sgn}(h)\abs{h}^{p-1}\mathcal{B}^{(\delta)}_\eps h \:dx\right)\:dv.
\end{split}
\end{equation}

Observing that 
$$\int_{\T^d}\mbox{sgn}(h)\abs{h}^{p-1} v\cdot\nabla_x h \:dx = \frac{1}{p}v\cdot\int_{\T^d} \nabla_x \left(\abs{h}^p\right) \:dx =0,$$

we deduce 

\begin{equation*}
\begin{split}
\frac{d}{dt}\norm{h}_{L^q_vL^p_x\left(\langle v \rangle ^k\right)} =& \norm{h}_{L^q_vL^p_x\left(\langle v \rangle ^k\right)}^{1-q} 
\\&\times\frac{1}{\eps^2}\int_{\R^d}\left(1+\abs{v}^k\right)^q \norm{h}^{q-p}_{L^p_x}\left(\int_{\T^d}\mbox{sgn}(h)\abs{h}^{p-1}\mathcal{B}^{(\delta)}_1 h \:dx\right)\:dv.
\end{split}
\end{equation*}
We can therefore use Lemma $\ref{lem:controlB1}$ which leads to
\begin{equation}\label{noderivativefinal}
\frac{d}{dt}\norm{h}_{L^q_vL^p_x\left(\langle v \rangle ^k\right)} \leq -\frac{1}{\eps^2}\left[1-\Lambda_{k-\gamma/q',q}(\delta)\right]\norm{h}^q_{L^q_vL^p_x\left(\langle v \rangle^k \nu^{1/q}\right)}\norm{h}^{1-q}_{L^q_vL^p_x\left(\langle v \rangle ^k\right)},
\end{equation}
We already noticed that $\Lambda_{k-1/q',q}(\delta)$ is strictly less than $1$ for $\delta$ smaller than some $\delta_{k,q}$ (see Remark $\ref{rem:lambdakq})$. Therefore, because $\nu (v) \geq \nu_0$ for all $v$ we have that for all $\delta$ smaller than $\delta_{k,q}$ the following holds,

$$\frac{d}{dt}\norm{h}_{L^q_vL^p_x\left(\langle v \rangle ^k\right)} \leq -\frac{\nu_0}{\eps^2}\left[1-\Lambda_{k-\gamma/q',q}(\delta)\right]\norm{h}_{L^q_vL^p_x\left(\langle v \rangle ^k\right)},$$

\par This concludes the proof of Proposition $\ref{prop:hypodissipativityBeps}$ for $\alpha=0$ and $1\leq p,q < +\infty$. The cases $p=\infty$ and $q=\infty$ are respectively dealt with by taking the limit $p\to \infty$ and $q\to \infty$ which is possible since $\delta_{k,q}$ is independent of $p$ and can be chosen to converge to a strictly positive constant when $q$ goes to $\infty$, thanks to the definition of $\Lambda_{k,q}(\delta)$.

\bigskip
\textbf{Step $2$: the case with $v$-derivatives}. Take $p,q$ in $[1,+\infty]$ and $\alpha=\beta=1$.
\par Since the $x$-derivative commutes with $\eqref{EDPdissipativity}$ the equation satisfied by $h$, we have that  $\eqref{noderivativefinal}$ holds for $x$-derivatives. Notice that $1-q \leq 0$ gives

\begin{equation}\label{dissipativityxderivative}
\begin{split}
&\frac{d}{dt}\left(\norm{h}_{L^q_vL^p_x\left(\langle v \rangle ^k\right)} + \norm{\nabla_xh}_{L^q_vL^p_x\left(\langle v \rangle ^k\right)} \right)
\\&\quad\quad\quad \leq -\frac{\nu_0^{1-1/q}}{\eps^2}\left[1-\Lambda_{k-\gamma/q',q}(\delta)\right]\left(\norm{h}_{L^q_vL^p_x\left(\langle v \rangle ^k\nu^{1/q}\right)}+\norm{\nabla_x h}_{L^q_vL^p_x\left(\langle v \rangle ^k\nu^{1/q}\right)}\right).
\end{split}
\end{equation}

\bigskip
Applying a $v$-derivatives to $\eqref{EDPdissipativity}$ yields
\begin{eqnarray*}
\partial_t \nabla_v h &=& \mathcal{B}^{(\delta)}_\eps(\nabla_v h) + \left(\nabla_v B^{(\delta)}_\eps\right)(h)
\\&=& \mathcal{B}^{(\delta)}_\eps(\nabla_v h) -\frac{1}{\eps}\nabla_x h +\mathcal{R}^{(\delta)}_\eps (h),
\end{eqnarray*}
where $\mathcal{R}^{(\delta)}_\eps (h) = \left(\nabla_v \mathcal{B}^{(\delta)}_{2,\eps}\right)(h)-\frac{1}{\eps^2}\nabla_v(\nu) h = \frac{1}{\eps^2}\mathcal{R}^{(\delta)}_1 (h)$.
\par From $\eqref{dissipativityxderivative}$, our computations in Step $1$ with $\delta \leq \delta_{k,q}$ and the following norm
$$\norm{h}_{W^{1,q}_vW^{1,p}_x\left(\langle v \rangle^k\right)_\eta} = \norm{h}_{L^q_vL^p_x\left(\langle v \rangle^k\right)} + \norm{\nabla_xh}_{L^q_vL^p_x\left(\langle v \rangle^k\right)} + \eta \norm{\nabla_vh}_{L^q_vL^p_x\left(\langle v \rangle^k\right)},$$
with $\eta >0$ to be fixed later, we obtain

\begin{equation*}
\begin{split}
& \frac{d}{dt}\norm{h}_{W^{1,q}_vW^{1,p}_x\left(\langle v \rangle^k\right)_\eta} 
\\ &\quad\leq -\frac{\nu_0^{1-1/q}}{\eps^2}\left[1-\Lambda_{k-\gamma/q',q}(\delta)\right]\left(\norm{h}_{L^q_vL^p_x\left(\langle v \rangle ^k\nu^{1/q}\right)}+\norm{\nabla_x h}_{L^q_vL^p_x\left(\langle v \rangle ^k\nu^{1/q}\right)}\right)
\\&\quad\quad -\eta\frac{\nu_0^{1-1/q}}{\eps^2}\left[1-\Lambda_{k-\gamma/q',q}(\delta)\right] \norm{\nabla_v h}_{L^q_vL^p_x\left(\langle v \rangle ^k\nu^{1/q}\right)}
\\ &\quad\quad -\frac{\eta}{\eps} \norm{\nabla_vh}_{L^q_vL^p_x\left(\langle v \rangle ^k\right)}^{1-q}\int_{\R^d}\left(\langle v\rangle^k\right)^q \norm{\nabla_vh}^{q-p}_{L^p_x}\left(\int_{\T^d}\mbox{sgn}(h)\abs{\nabla_vh}^{p-1}\nabla_xh \:dx\right)\:dv
\\ &\quad\quad +\frac{\eta}{\eps^2} \norm{\nabla_vh}_{L^q_vL^p_x\left(\langle v \rangle ^k\right)}^{1-q}\int_{\R^d}\left(\langle v\rangle^k\right)^q \norm{\nabla_vh}^{q-p}_{L^p_x}\left(\int_{\T^d}\mbox{sgn}(h)\abs{\nabla_vh}^{p-1}\mathcal{R}^{(\delta)}_1(h) \:dx\right)\:dv.
\end{split}
\end{equation*}

We take the absolute value and use H\"older inequality twice on the last two terms which makes the terms in $\nabla_vh$ disappear, and this gives

\begin{equation*}
\begin{split}
& \frac{d}{dt}\norm{h}_{W^{1,q}_vW^{1,p}_x\left(\langle v \rangle^k\right)_\eta} 
\\ &\quad\leq -\frac{\nu_0^{1-1/q}}{\eps^2}\left[1-\Lambda_{k-\gamma/q',q}(\delta)\right]\left(\norm{h}_{L^q_vL^p_x\left(\langle v \rangle ^k\nu^{1/q}\right)} +\eta \norm{\nabla_v h}_{L^q_vL^p_x\left(\langle v \rangle ^k\nu^{1/q}\right)}\right)
\\&\quad\quad + \frac{1}{\eps^2}\left(\eps\eta\nu_0^{-1/q} -\nu_0^{1-1/q}\left[1-\Lambda_{k-\gamma/q',q}(\delta)\right] \right)\norm{\nabla_x h}_{L^q_vL^p_x\left(\langle v \rangle ^k\nu^{1/q}\right)} 
\\&\quad\quad+ \frac{\eta}{\eps^2}\norm{\mathcal{R}^{(\delta)}_1 (h)}_{L^q_vL^p_x\left(\langle v \rangle^k\right)}.
\end{split}
\end{equation*}

\bigskip
One can find in \cite{GMM} proof of Lemma $4.14$ case $(W2)$ and $(W3)$ the following estimate
$$\norm{\mathcal{R}^{(\delta)}_1 (h)}_{L^q_vL^p_x\left(\langle v \rangle^k\right)} \leq C_\delta \norm{h}_{L^q_vL^p_x\left(\langle v \rangle^k\nu^{1/q}\right)},$$
where $C_\delta >0$ is a constant only depending on $\delta$.
\par Because $\eps \leq 1$, this latter estimates yields
\begin{equation}\label{vderivativefinal}
\begin{split}
\frac{d}{dt}\norm{h}_{W^{1,q}_vW^{1,p}_x\left(\langle v \rangle^k\right)_\eta} \leq & \frac{1}{\eps^2}\left(C_\delta \eta - \nu_0^{1-1/q}\left[1-\Lambda_{k-\gamma/q',q}(\delta)\right]\right)\norm{h}_{L^q_vL^p_x\left(\langle v \rangle ^k\nu^{1/q}\right)}
\\& + \frac{1}{\eps^2}\left(\eta\nu_0^{-1/q} -\nu_0^{1-1/q}\left[1-\Lambda_{k-\gamma/q',q}(\delta)\right] \right)\norm{\nabla_x h}_{L^q_vL^p_x\left(\langle v \rangle ^k\nu^{1/q}\right)} 
\\& -\eta \frac{\nu_0^{1-1/q}}{\eps^2}\left[1-\Lambda_{k-\gamma/q',q}(\delta)\right]\norm{\nabla_v h}_{L^q_vL^p_x\left(\langle v \rangle ^k\nu^{1/q}\right)},
\end{split}
\end{equation}
which concludes the proof if we take $\eta$ small enough in terms of $\delta$, for $\delta\leq \delta_{k,q}$.

\bigskip
The case where $1< \alpha = \beta$ is dealt with in the same way with the norm
$$\norm{h}_{W^{\alpha,q}_vW^{\alpha,p}_x\left(\langle v \rangle^k\right)_\eta} = \sum\limits_{0\leq \abs{j}+\abs{l} \leq \alpha} \eta^{\abs{j}} \norm{\partial_l^j h}_{L^q_vL^p_x\left(\langle v \rangle^k\right)},$$
with $\eta$ small enough in terms of $\delta$.
\end{proof}


\subsection{Estimates on the iterated convolution product, assumption $(A2)(iii)$}\label{subsec:iteratedconvolution}

In order to use Theorem $\ref{theo:extension}$, it remains to show that our equation $\eqref{LBE1}$ satisfies hypothesis $(A2)(iii)$, that is we need to control the iterated quantities $T_l:=\left(\mathcal{A}^{(\delta)}_\eps S_{\mathcal{B}^{(\delta)}_\eps}\right)^{(*l)}$ for some $l$ in $\N$. The following proposition describes such controls when $p=1$.

\bigskip
\begin{prop}\label{prop:T1T2L1}
Consider $k>k_q^*$, defined by $\eqref{kq*}$, and $s$ in $\N$ .
\\For any $\delta$ in $(0,\delta_{k,q}]$ and any $\lambda_0'$ in $(0,\lambda_0)$ ($\delta_{k,q}$ and $\lambda_0$ defined in Proposition $\ref{prop:hypodissipativityBeps}$), there exists $C_1 = C_1(\lambda_0',\delta)>0$ and $R=R(\delta)>0$ such that for any $t\geq 0$,
$$\forall n \in \N, \quad \emph{\mbox{supp}}\:T_n(t)h \subset K := B(0,R)$$
and
\begin{eqnarray}
&&\forall s\geq 1,\quad \norm{T_1(t)h}_{W^{s+1,1}_{x,v}(K)} \leq C_1 \frac{e^{-\frac{\lambda_0'}{\eps^2}t
}}{\eps^2 t}\norm{h}_{W^{s+1,1}_{v}W^s_x(\langle v \rangle^k)}, \label{controlT1}
\\ &&\forall s\geq 0,\quad \norm{T_2(t)h}_{W^{s+1/2,1}_{x,v}(K)} \leq C_1 \frac{e^{-\frac{\lambda_0'}{\eps^2}t
}}{\eps^{4}}\norm{h}_{W^{s,1}_{x,v}(\langle v \rangle^k)}. \label{controlT2}
\end{eqnarray}
\end{prop}

\begin{proof}[Proof of Proposition $\ref{prop:T1T2L1}$]
\par Most of the proof is an adaptation of \cite{GMM} proof of Lemma $4.19$ to keep track of the dependencies on $\eps$. We will refer to it when we are using some of its computations.

\bigskip
\textbf{Control of $T_1(t)h$: } The $x$-derivatives commutes with $T_1(t)$ and therefore it is enough to consider $h$ in $W^{s,1}_vW^{1,1}_x(\langle v \rangle^k)$, with $s \geq 1$, and to control $\norm{T_1(t)h}_{W^{s+1,1}_{v}W_x^{1,1}(K)}$. This gives

\begin{equation}\label{controlT1norm}
\norm{T_1(t)h}_{W^{s+1,1}_{v}W_x^{1,1}(K)} \leq \norm{T_1(t)h}_{W^{s+1,1}_{v}L_x^{1}(K)} +\norm{\nabla_xT_1(t)h}_{W^{s+1,1}_{v}L_x^{1}(K)}.
\end{equation}

\bigskip
The first term is easily dealt with thanks to the estimate on $\mathcal{A}^{(\delta)}_\eps$, Proposition $\ref{prop:controlAeps}$, and the dissipativity property of $\mathcal{B}^{(\delta)}_\eps$, Proposition $\ref{prop:hypodissipativityBeps}$,

\begin{equation}\label{controlT1term1}
\norm{T_1(t)h}_{W^{s+1,1}_{v}L_x^{1}(K)} = \norm{A^{(\delta)}_\eps S_{\mathcal{B}^{(\delta)}_\eps} h}_{W^{s+1,1}_{v}L_x^{1}(K)} \leq \frac{C}{\eps^2}e^{-\frac{\lambda_0}{\eps^2}t}\norm{h}_{L^{1}_vL^{1}_x(\langle v \rangle^k)}.
\end{equation}

\bigskip
For the second term, define $f(t) = S_{\mathcal{B}^{(\delta)}_\eps} h$ and 
\begin{equation}\label{defDt}
D_t = \eps^{-1}t\nabla_x + \nabla_v.
\end{equation}
By direct computations we have that

$$\eps^{-1}t\nabla_xT_1(t)h = \mathcal{A}^{(\delta)}_\eps(D_tf) -  \left(\nabla_v\mathcal{A}^{(\delta)}_\eps\right)f,$$
which leads to, by Proposition $\ref{prop:controlAeps}$,

\begin{equation}\label{controlT1term2a}
\eps^{-1}t \norm{\nabla_xT_1(t)h}_{W^{s+1,1}_{v}L_x^{1}(K)} \leq \frac{C}{\eps^2}\left[\norm{D_tf}_{L^1_{x,v}\left(\langle v \rangle^k\right)}+ \norm{f}_{L^1_{x,v}\left(\langle v \rangle^k\right)}\right].
\end{equation}

The dissipativity property of $\mathcal{B}^{(\delta)}_\eps$, in particular $\eqref{noderivativefinal}$ with $q=1$, yields 

\begin{equation}\label{controlf}
\frac{d}{dt}\norm{f}_{L^1_{x,v}\left(\langle v \rangle^k\right)} \leq -\frac{\lambda_0}{\eps^2}\norm{f}_{L^1_{x,v}\left(\langle v \rangle^k\nu\right)}.
\end{equation}

Direct computations yields
$$\partial_t\left(D_tf\right) = \mathcal{B}^{(\delta)}_\eps\left(D_t f\right) + \frac{1}{\eps^2}\mathcal{J}^{(\delta)}f,$$
where 
\begin{equation}\label{Jdelta}
\mathcal{J}^{(\delta)} = \nabla_v \left(\mathcal{B}^{(\delta)}_1(\cdot)\right) - \mathcal{B}^{(\delta)}_1\left(\nabla_v (\cdot)\right)
\end{equation}
is independent of $\eps$ and satisfies (see \cite{GMM} proof of Lemma $4.19$) for all $g$ in $L^1_v\left(\langle v \rangle^k\nu\right)$
$$\norm{\mathcal{J}^{(\delta)}g}_{L^1_v\left(\langle v \rangle^k\right)} \leq C_\delta \norm{g}_{L^1_v\left(\langle v \rangle^k\nu\right)}.$$
In the same way as proof of Proposition $\ref{prop:hypodissipativityBeps}$ we obtain

\begin{equation}\label{Dtf}
\frac{d}{dt}\norm{D_tf}_{L^1_{x,v}\left(\langle v \rangle^k\right)} \leq  -\frac{\lambda_0}{\eps^2}\norm{D_tf}_{L^1_{x,v}\left(\langle v \rangle^k\nu\right)} +\frac{C_\delta}{\eps^2}\norm{f}_{L^1_v\left(\langle v \rangle^k\nu\right)}.
\end{equation}
We then consider $\lambda_0'$ in $(0,\lambda_0)$ and define $\eta = (\lambda_0-\lambda_0')/C_\delta$. We compute, with $\eqref{controlf}$,

$$\frac{d}{dt}\left[e^{\frac{\lambda_0'}{\eps^2}t}\left(\eta\norm{D_tf}_{L^1_{x,v}\left(\langle v \rangle^k\right)} + \norm{f}_{L^1_{x,v}\left(\langle v \rangle^k\right)} \right)\right]\leq 0,$$

and thus

\begin{equation}\label{normDtff}
\norm{D_tf}_{L^1_{x,v}\left(\langle v \rangle^k\right)} + \norm{f}_{L^1_{x,v}\left(\langle v \rangle^k\right)}  \leq \eta^{-1}e^{-\frac{\lambda_0'}{\eps^2}t}\norm{h}_{W^{1,1}_vL^{1}_x\left(\langle v\rangle^k\right)}.
\end{equation}

To conclude we plug $\eqref{normDtff}$ into $\eqref{controlT1term2a}$ and we combine it with $\eqref{controlT1term1}$ into $\eqref{controlT1norm}$. This yields, because $s\geq 1$,

$$\norm{T_1(t)h}_{W^{s+1,1}_{v}W^{1,1}_x(K)} \leq C \frac{e^{-\frac{\lambda_0'}{\eps^2}t}}{\eps^2t}\norm{h}_{W^{s,1}_{v}L^{1}_x(\langle v \rangle^k)},$$
which implies the expected result $\eqref{controlT1}$ because $T_1(t)$ commutes with $x$-derivatives.

\bigskip
\textbf{Control of $T_2(t)h$: } For $s \geq 0$ we can interpolate (for interpolation theory in Sobolev spaces see \cite{BerLof} Chapters $6$) between $\eqref{controlT1term1}$ and $\eqref{controlT1}$ to get

$$\norm{T_1(t)h}_{W^{s+1/2,1}_{x,v}(K)} \leq C \frac{e^{-\frac{\lambda_0'}{\eps^2}t}}{\eps^{2}\sqrt{t}}\norm{h}_{W^{s,1}_{v}L^{1}_x(\langle v \rangle^k)}.$$

Then, we firstly use  the inequality above and secondly $\eqref{controlT1term1}$ to obtain

\begin{eqnarray*}
\norm{T_2(t)h}_{W^{s+1/2,1}_{x,v}(K)} &\leq& \int_0^t \norm{T_1(t-s)T_1(s)h}_{W^{s+1/2,1}_{x,v}(K)}\:ds
\\&\leq& \frac{C}{\eps^{4}}e^{-\frac{\lambda_0't}{\eps^2}}\left(\int^t_0\frac{e^{-\frac{\lambda_0-\lambda_0'}{\eps^2}s}}{\sqrt{t-s}}\:ds\right)\norm{h}_{W^{s,1}_{v}L^{1}_x(\langle v \rangle^k)},
\end{eqnarray*}
which is the expected result $\eqref{controlT2}$.

\end{proof}

\bigskip
The aim is to link our space $L^q_vL^p_x\left(\langle v \rangle^k\right)$ to the space $H^s_{x,v}\left(\mu^{-1/2}\right)$. We thus state the following control on the iterated convolution in the case $p=2$.

\bigskip
\begin{prop}\label{prop:T2L2}
Consider $k>k_q^*$, defined by $\eqref{kq*}$, and $s$ in $\N$ .
\\For any $\delta$ in $(0,\delta_{k,q}]$ there exists $C_2 = C_2(\delta)>0$ and $R=R(\delta)>0$ such that for any $t\geq 0$,
$$\forall n \in \N, \quad \emph{\mbox{supp}}\:T_n(t)h \subset K := B(0,R)$$
and
\begin{equation}\label{controlT2L2}
\forall s\geq 0,\quad \norm{T_2(t)h}_{H^{s+1/2}_{x,v}(K)} \leq \frac{C_T}{\eps^{5/2}} \norm{h}_{H^{s}_{x,v}(\langle v \rangle^k)}.
\end{equation}
\end{prop}

\begin{proof}[Proof of Proposition $\ref{prop:T2L2}$]
Consider $h$ in $W^{s,2}_{x,v}(\langle v \rangle^k)$, $s$ in $\N$.
\par This Proposition is easier than when $p=1$ because there exists velocity averaging lemmas in this framework, as discussed in \cite{GMM} Remark $4.21$. The $x$-derivative commutes with $T_1$ and therefore we suppose there is no derivative in space. 
\par Define $f(t) = S_{\mathcal{B}^{(\delta)}_\eps}(t)(h)$  so that $f$ is solution to the kinetic equation

$$\partial_t f + \frac{1}{\eps} v\cdot \nabla_x f = s_\eps(t,x,v),$$

with $s_\eps(t,x,v) = -\eps^{-2}\nu f + \eps^{-2}\mathcal{B}^{(\delta)}_{2,\eps}f$.
\par Let $j$ be a multi-index such that $\abs{j} \leq s$. We apply $\partial^j_0$ to the latter equation, which gives
\begin{equation}\label{kineticp=2}
\partial_t \left(\partial^j_0f\right) + \frac{1}{\eps} v\cdot \nabla_x \left(\partial^j_0 f\right) = \partial^j_0s_\eps(t,x,v) + \frac{1}{\eps}\sum\limits_{\abs{i}+\abs{l} = \abs{j}}a_{i,l}\partial^i_l f,
\end{equation}
where $a_{i,j} $ are non-negative numbers.
\par A classical averaging lemma (see \cite{BouchDes} Lemma $1$ and \cite{BoudDes} in which we emphasize the dependencies in $\eps$) reads, for $\eqref{kineticp=2}$ with $\partial^j_0f(0,x,v) = \partial^j_0h(x,v)$, for all $\psi$ in $\mathcal{D}\left(\R^d\right)$

\begin{equation}\label{averaginglemma}
\begin{split}
& \norm{\int_{\R^d}\partial^j_0f(t,x,v)\psi(v)\:dv}_{L^2_t\left(H^{1/2}_x\right)} 
\\ &\quad\quad\quad\leq \frac{C}{\sqrt{\eps}}\left(\norm{\partial^j_0h(x,v)}_{L^2_{x,v}}+\norm{\partial^j_0s_\eps}_{L^2_{t,x,v}}+\frac{1}{\eps}\norm{\sum\limits_{\abs{i}+\abs{l} = \abs{j}}a_{i,l}\partial^i_l f}_{L^2_{t,x,v}}\right).
\end{split}
\end{equation}

\bigskip
We use \cite{GMM}, Lemmas $4.4$ and $4.7$, in order to bound the terms involving $\mathcal{B}^{(\delta)}_{2,\eps} =\eps^{-2}\mathcal{B}^{(\delta)}_{2,1}$ we have that

$$\norm{s_\eps}_{H^s_{x,v}\left(\langle v \rangle^k \right)} \leq \frac{1}{\eps^2}\norm{s_1}_{H^s_{x,v}\left(\langle v \rangle^k \right)} \leq \frac{C}{\eps^2}\norm{f}_{H^s_{x,v}\left(\langle v \rangle^k \nu\right)} \leq \frac{C}{\eps^2}e^{-\frac{\lambda_0}{\eps^2}t}\norm{h}_{H^s_{x,v}\left(\langle v \rangle^k \nu\right)},$$
where the last inequality comes from the hypodissipativity properties of $S_{\mathcal{B}_\eps}(t)$, see Proposition $\ref{prop:hypodissipativityBeps}$.
\par Using the dissipativity properties of $S_{\mathcal{B}_\eps}(t)$ one more time we deduce that

\begin{equation}\label{finalineqT2L2}
\norm{T_1(t)h}_{L^2_t\left(H^{s+1/2}_{x,v}\left(\langle v \rangle^k \right)\right)} \leq \frac{C}{\eps^{5/2}} \norm{h}_{H^s_{x,v}\left(\langle v \rangle^k \nu\right)}.
\end{equation}
To conclude we notice that $\int_0^t T_1(t-s)T_1(s)\:ds$ is a continuous linear operator on the Hilbert space $H^{s+1/2}_{x,v}(K)$ and thus we can see it as an element of $H^{s+1/2}_{x,v}(K)$ by Riesz's representation theorem. Hence, thanks to Cauchy-Schwartz,

\begin{eqnarray*}
 \norm{T_2(t)h}_{H^{s+1/2}_{x,v}(K)} &=&  \norm{\left(\int_0^t T_1(t-s)T_1(s)\:ds\right)(h)}_{H^{s+1/2}_{x,v}(K)}
 \\&\leq& \norm{h}_{H^s_{x,v}\left(\langle v \rangle^k \nu\right)} \int_0^t \norm{T_1(t-s)T_1(s)}_{\mathscr{B}\left(H^s_{x,v}\left(\langle v \rangle^k \nu\right),H^{s+1/2}_{x,v}(K)\right)}\:ds
 \\&\leq& \norm{h}_{H^s_{x,v}\left(\langle v \rangle^k \nu\right)} \left(\int_0^t \norm{T_1(t-s)}^2_{\mathscr{B}\left(H^s_{x,v}(K),H^{s+1/2}_{x,v}(K)\right)}\:ds\right)^{1/2}
\\&& \quad\quad\quad \times\left(\int_0^t \norm{T_1(s)}^2_{\mathscr{B}\left(H^s_{x,v}\left(\langle v \rangle^k \nu\right),H^s_{x,v}(K)\right)}\:ds\right)^{1/2}
 \\&\leq& \norm{h}_{H^s_{x,v}\left(\langle v \rangle^k \nu\right)} \frac{C}{\eps^{5/2}} \left(\int_0^t \frac{C_A}{\eps^2}e^{-\frac{\lambda_0'}{\eps^2}s}\:ds\right)^{1/2}
\\&\leq& \frac{C}{\eps^{5/2}}\norm{h}_{H^s_{x,v}\left(\langle v \rangle^k \nu\right)},
 \end{eqnarray*}
where we used Proposition $\ref{prop:controlAeps}$ and the fact that $S_{\mathcal{B}^{(\delta)}_\eps}$ is a contraction semigroup on $H^s_{x,v}$ with spectral gap $\lambda_0'/\eps^2$.

\end{proof}


\subsection{Proof of Theorem $\ref{theo:linLqLp}$}\label{subsec:linproof}

As we explained it in Section $\ref{subsec:strategylinLpLq}$, the proof of Theorem $\ref{theo:linLqLp}$ is direct from the application of Theorem $\ref{theo:extension}$. This theorem is clearly applicable in our case and we emphasize it through the extreme case of no derivative in space or velocity variables.
\par Indeed, we consider $s$ in $\N$ to be chosen big enough later. We define $\mathcal{E} = L^q_vL^p_x\left(\langle v \rangle^k\right)$ and $E = H^s_{x,v}\left(\mu^{-1/2}\right)$ and we have $E \subset \mathcal{E}$ for $s$ big enough (dense with continuous embedding). Indeed, in the case $q \geq 2$ and $p \geq 2$, standard Sobolev embeddings (see \cite{Br} Section $IX.3.$) imply $E \subset L^q_{v}L^p_x\left(\mu^{-1/2}\right)$. In the case $p <2$ we have, on the torus, $L^2_x \subset L^p_x$ and $H^s_x \subset L^2_x$ by the same Sobolev embeddings. Finally, in the case $q < 2$ we have that $L^2_v\left(\mu^{-1/2}\right) \subset L^q_v\left(\langle v \rangle^k\right)$ (it can be done by a mere Cauchy-Schwarz inequality) and the same Sobolev embeddings give $H^s_v\left(\mu^{-1/2}\right) \subset L^2_v\left(\mu^{-1/2}\right)$.

\bigskip
 On the torus we have the following embedding: $L^p_x\subset L^1_x$. Thanks to Proposition $\ref{prop:controlAeps}$ and Proposition $\ref{prop:hypodissipativityBeps}$ we obtain (same arguments as $\eqref{controlT1term1}$)

\begin{equation}\label{A2iii1ststep}
\norm{T_1(t)h}_\mathcal{E} \leq C\norm{A^{(\delta)}_\eps S_{\mathcal{B}^{(\delta)}_\eps} h}_{L^q_vL_x^{1}(K)} \leq \frac{C}{\eps^2}e^{-\frac{\lambda_0}{\eps^2}t}\norm{h}_{L^{1}_vL^{1}_x(\langle v \rangle^k)}.
\end{equation}
We therefore define $\mathcal{E}_2=L^{1}_vL^{1}_x(\langle v \rangle^k)$.
\par Then we define by $\mathcal{E}_j = W^{(j-2)/2,1}_{x,v}(\langle v \rangle^k)$ for $j$ from $2$ to $m$ with $m$ big enough such that $W^{(m-1)/2,1}_{x,v}(\langle v \rangle^k)\subset L^2_{x,v}(\langle v \rangle^k)$. Then we denote $\mathcal{E}_j = H^{(j-m-1)/2}_{x,v}(\langle v \rangle^k)$ up $J-1$  where $H^{(J-m-2)/2}_{x,v}(\langle v \rangle^k)\subset E$.
\par Point $(A1)$ of Theorem $\ref{theo:extension}$ is satisfied thanks to \cite{BM} and \cite{Bri1} Theorem $2.1$ (with the norm of Theorem $2.4$), point $(A2)(i)$ by Proposition $\ref{prop:hypodissipativityBeps}$ and point $(A2)(ii)$ by Proposition $\ref{prop:controlAeps}$. Finally, point $(A2)(iii)$ is given by $\eqref{A2iii1ststep}$ for $\mathcal{E}$ and $\mathcal{E}_1$, then by Proposition $\ref{prop:T1T2L1}$ $\eqref{controlT2}$ up to $\mathcal{E}_m$ and by Proposition $\ref{prop:T2L2}$ from $\mathcal{E}_m$ to $\mathcal{E}_J$ and $E$.

\section{An \textit{a priori} estimate for the full perturbed equation: proof of Theorem $\ref{theo:cauchyexpodecay}$}\label{sec:cauchyexpodecay}

In this section we work in $W^{\alpha,1}_vH^{\beta}_x\left(\langle v \rangle^k\right)$ or in $W^{\alpha,1}_vW^{\beta,1}_x\left(\langle v \rangle^k\right)$, with $\alpha \leq \beta$ on the full perturbed Boltzmann equation
$$\partial_th  = \mathcal{G}_\eps(h) + \frac{1}{\eps}Q(h,h).$$


\subsection{Description of the problem and notations}\label{subsec:descriptionpb}

When $\eps=1$, the linear part $\mathcal{G}_\eps$ has the same order of magnitude than the bilinear term $Q$ in the linearized Boltzmann equation $\eqref{LBE}$. In this case, Theorem $\ref{theo:linLqLp}$ suffices to obtain existence and exponential decay since the contraction property of the semigroup $S_{\mathcal{G}_1}$ controls the bilinear part for small initial data (see \cite{GMM}).
\par In the general case, $S_{\mathcal{G}_\eps}$ only generates a semigroup with a spectral gap of order $1$, insufficient to control $\eps^{-1}Q$. However, \cite{Gu4}\cite{Bri1} show that a careful study of $\eps^{-1}Q$ compared to $G_\eps$ yields existence and exponential decay of solutions to $\eqref{LBE}$ in $H^s_{x,v}\left(\mu^{-1/2}\right)$ for $s$ large enough (see Theorem $\ref{theo:adaptedBri1}$ for  an adapted version of this result). Our strategy is to use the same kind of ideas as when we extended the semigroup properties from $H^\beta_{x,v}\left(\mu^{-1/2}\right)$ to $W^{\alpha,1}_vH^{\beta}_x\left(\langle v \rangle^k\right)$ and $W^{\alpha,1}_vW^{\beta,1}_x\left(\langle v \rangle^k\right)$ but including the bilinear term. Namely, we shall decompose the partial differential equation $\eqref{LBE}$ into a system of partial differential equations from $W^{\alpha,1}_vH^{\beta}_x\left(\langle v \rangle^k\right)$ or $W^{\alpha,1}_vW^{\beta,1}_x\left(\langle v \rangle^k\right)$  to $H^\beta\left(\mu^{-1/2}\right)$ and use the perturbative estimates of \cite{Bri1}.

\bigskip
As noticed in Remark $2.16$ of \cite{GMM}, Theorem $\ref{theo:extension}$ extending the semigroup generated by $\mathcal{G_\eps}$  in $H^s\left(\mu^{-1/2}\right)$ to $L^1_vL^\infty_x\left(\langle v \rangle^k\right)$ can be interpreted as a decomposition of
$$\partial_tf = \mathcal{G_\eps}f,$$
into a system of partial differential equations, involving operators $\mathcal{G}_\eps = \mathcal{A}_\eps +\mathcal{B}_\eps$ (defined in Section $\ref{subsec:lin_decompositionL}$), with $f = f^1+\dots+f^{J}$ satisfying
\begin{itemize}
\item $f^1$ is in $L^1_vL^\infty_x\left(\langle v \rangle^k\right)$ and $f^1_{in} = f_{in}$ in $\emph{\mbox{Ker}}(\mathcal{G}_\eps)^\bot$,
\item for all $2\leq j\leq  J-1$, $f^j$ is in $\mathcal{E}_j$ and $f^j_{in} = 0$,
\item $f^{J}$ is in $H^s\left(\mu^{-1/2}\right)$, $f^J_{in} = 0$ and in that space we can use the contraction property of $S_{\mathcal{G}_\eps}$.
\end{itemize}
We will decompose the linearized Boltzmann equation in a similar way than the one explained above. We shall define a sequence of spaces $\left(\mathcal{E}_j\right)_{1\leq j\leq J}$. In each space $\mathcal{E}_j$, $1\leq j\leq J-1$, a piece of the bilinear term, of order $\eps^{-1}$, will be added and controlled by the dissipativity property of $\mathcal{B}^{(\delta)}_\eps$, of order $\eps^{-2}$. Contrary to the study in the linear case, the bilinear operator generates terms involving functions in all the spaces $\mathcal{E}_j$ which have to be compared and controlled. This imposes to construct $\left(\mathcal{E}_j\right)_{1\leq j\leq J}$ as a nested sequence.
\par The difficult part of the linear operator, namely $\mathcal{A}^{(\delta)}_\eps$, enjoys a regularising effect and could therefore be treated in more regular spaces. Of course, our decomposition will be much easier since we solely want to go from an exponential weight into a polynomial weight Sobolev spaces, without losing any derivatives in $x$ or $v$.

\bigskip
In order to shorten notations we define, for $p=1,2$ and $k$ to be defined later,
\begin{equation}\label{definitionspaces}
\mathcal{E}^{p} = W^{\alpha,1}_vW^{\beta,p}_x\left(\langle v \rangle^k\right) \quad\mbox{and}\quad E = H^\beta_{x,v}\left(\mu^{-1/2}\right).
\end{equation}
We take $h_{in}$ in $\mathcal{E}^{p}$ and we decompose the partial differential equation,
$$\partial_t h = \mathcal{G}_\eps(h) + \frac{1}{\eps}Q(h,h)= \mathcal{A}^{(\delta)}_\eps(h) +\mathcal{B}^{(\delta)}_\eps(h)+ \frac{1}{\eps}Q(h,h)$$
into an equivalent system of partial differential equations for the following decomposition
\begin{equation}\label{decompositionh}
h(t,x,v) = h^0(t,x,v) + h^1(t,x,v),
\end{equation}
with
\begin{enumerate}
\item In $\mathcal{E}^{p}$, $h^{0}_{t=0} = h_{in}$ and
\begin{equation}\label{decompositionh0}
\partial_t h^{0} = \mathcal{B}^{(\delta)}_\eps(h^{0}) + \frac{1}{\eps} Q(h^0,h^0) + \frac{2}{\eps} Q\left(h^0,h^1\right), 
\end{equation}
\item In $E$, $h^{1}_{t=0} = 0$ and
\begin{equation}\label{decompositionh1}
\partial_t h^1 =  \mathcal{G}_\eps(h^{1}) + \frac{1}{\eps} Q(h^{1},h^{1}) + \mathcal{A}^{(\delta)}_\eps(h^{0}).
\end{equation}
\end{enumerate}

\bigskip
The aim of this Section is to establish the following estimate of solutions to the system $\eqref{decompositionh0}-\eqref{decompositionh1}$.

\bigskip
\begin{theorem}\label{theo:aprioridecomposition}
Let $p=1$ or $p=2$. There exist $\beta_0$ in $\N$ and $\eps_d$ in $(0,1]$ depending on $d$ and the kernel of the Boltzmann operator such that:
\\For all $\beta\geq \beta_0$, for any $\delta$ in $(0,\delta_{k,1}]$ and any $\lambda_0'$ in $(0,\lambda_0)$ ($\delta_{k,1}$ and $\lambda_0$ defined in Proposition $\ref{prop:hypodissipativityBeps}$) there exist $C_\beta$, $\eta_\beta> 0$ such that for any $0< \eps \leq \eps_d$ and $h_{in}$ in $\mathcal{E}^{p}$,
\\\noindent if
\begin{enumerate}
\item[(i)] $\norm{h_{in}}_{\mathcal{E}^{p}} \leq \eta_\beta,$
\item[(ii)] $(h^0,h^1)$ is solution to the system $\eqref{decompositionh0}-\eqref{decompositionh1}$,
\end{enumerate}
\noindent then
$$\norm{h^0+h^1}_{\mathcal{E}^{p}} \leq C_\beta \norm{h_{in}}_{\mathcal{E}^{p}} e^{-\lambda_0't}.$$
The constants $C_\beta$ and $\eta_\beta$ are constructive and depends only on $\beta$, $d$, $\delta$, $\lambda_0'$ and the kernel of the Boltzmann operator.

\end{theorem}
\bigskip

\begin{remark}[Link with Theorem $\ref{theo:cauchyexpodecay}$]
The existence and uniqueness for the perturbed Boltzmann equation $\eqref{LBE}$ in $\mathcal{E}^{p}$ has been proved for $\eps=1$, that is equivalent of $\eps$ fixed with constant depending on it, in \cite{GMM} Theorems $5.3$ and $5.5$ respectively. The constants, as well as the smallness assumption on the initial data, in the theorem above are independent of $\eps$ and therefore this $\textit{a priori}$ result combined with existence and uniqueness developed in \cite{GMM} and in \cite{Bri1} implies the existence and uniqueness independently of $\eps$ which is Theorem $\ref{theo:cauchyexpodecay}$.
\end{remark}
\bigskip

The next subsections deal with the estimates one can get for solutions to the system $\eqref{decompositionh0}-\eqref{decompositionh1}$. We study each of them independently and the \textit{a priori} exponential decay will be a straightforward application of these results together with a maximum principle argument.
\par Section $\ref{subsec:h0}$ focuses on the \textit{a priori} study of the equation in $\mathcal{E}^{p}$. Section $\ref{subsec:h1}$ deals with $\eqref{decompositionh1}$ in $E$. Finally, Section $\ref{subsec:proofapriori}$ gathers the previous results to prove Theorem $\ref{theo:aprioridecomposition}$.


\subsection{Study of equation $\eqref{decompositionh0}$ in $\mathcal{E}$}\label{subsec:h0}

In this section we prove the following general proposition about the equation taking place in $\mathcal{E}^{p} = W^{\alpha,1}_vW^{\beta,p}_x\left(\langle v \rangle^k\right)$, for $p=1$ or $p=2$. We define the shorthand notation 
$$\mathcal{E}^{p}_\nu = W^{\alpha,1}_vW^{\beta,p}_x\left(\langle v \rangle^k\nu\right).$$

\bigskip
\begin{prop}\label{prop:h0}
Let $p=1$ or $p=2$ and $0<\eps \leq 1$. Let $k>k^*_1 =2$, $\beta > 2d/p$.
\\Let $h_{in}$ be in $\mathcal{E}^{p}$ and $h^1$ in $\mathcal{E}^{p}_\nu$.
\\For any $\delta$ in $(0,\delta_{k,1}]$ and any $\lambda_0'$ in $(0,\lambda_0)$ ($\delta_{k,1}$ and $\lambda_0$ defined in Proposition $\ref{prop:hypodissipativityBeps}$) there exist $\eta_0> 0$ such that
\\if
\begin{enumerate}
\item[(i)] $\norm{h_{in}}_{\mathcal{E}^{p}} \leq \eta_0\:, \quad\norm{h^1}_{\mathcal{E}^{p}_\nu} \leq \eta_0,$
\item[(ii)] $h^0$ satisfies $h^0_{t=0} = h_{in}$ and  is solution to 
$$\partial_t h^0 = \mathcal{B}^{(\delta)}_\eps(h^0) +\frac{1}{\eps} Q(h^0,h^0) + \frac{2}{\eps} Q\left(h^0,h^1\right),$$
\end{enumerate}

\noindent then
$$\norm{h^0}_{\mathcal{E}^{p}} \leq e^{-\frac{\lambda_0'}{\eps^2}t}\norm{h_{in}}_{\mathcal{E}^{p}}.$$
The constant $\eta_0$ is constructive and depends only on $\delta$, $\lambda_0'$ and the kernel of the Boltzmann operator.
\end{prop}
\bigskip

We need to control the bilinear term $Q$, which is given by the following lemma. 

\bigskip
\begin{lemma}\label{lem:controlbilinearlargeE}
For all $p=1,2$ and $\alpha$, $\beta$ in $\N$ such that $\beta>2d/p$, there exists $C_{\beta,p}>0$ such that all $f$ and $g$ 

\begin{equation*}
\norm{Q(f,g)}_{\mathcal{E}^{p}} \leq C_{\beta,p}\left(\norm{g}_{ \mathcal{E}^{p}_\nu}\norm{f}_{\mathcal{E}^{p}} + \norm{g}_{\mathcal{E}^{p}}\norm{f}_{\mathcal{E}^{p}_\nu} \right).
\end{equation*}
\end{lemma}
\bigskip

This lemma has been proved in Lemma $5.16$ in \cite{GMM}, which is adapted from interpolation results in \cite{Ark1} or duality arguments as in \cite{MouVi} Theorem $2.1$.

\bigskip
\begin{proof}[Proof of Proposition $\ref{prop:h0}$]

Consider $\delta$ in $(0,\delta_{k,1}]$ and $\lambda_0'$ in $(0,\lambda_0)$.  Take $p=1$ or $p=2$ and $\beta>2d/p$.
\par We have that
$$\partial_t h^0 = \mathcal{B}^{(\delta)}_\eps(h^0) + \frac{1}{\eps} Q(h^0,h^0) + \frac{2}{\eps} Q\left(h^0,h^1\right).$$
Thanks to the dissipativity of property of $\mathcal{B}^{(\delta)}_\eps$, more precisely the proof of Lemma $\ref{prop:hypodissipativityBeps}$, we have

\begin{equation*}
\begin{split}
&\frac{d}{dt}\norm{h^0}_{\mathcal{E}^{p}} \leq -\frac{\lambda_0}{\eps^2\nu_0} \norm{h^0}_{\mathcal{E}^{p}_\nu}+ \frac{1}{\eps}\abs{\langle Q(h^0,h^0)+2Q(h^0,h^1),h^0\rangle_{\mathcal{E}^{p}}}
\\&\leq -\frac{\lambda_0}{\eps^2\nu_0} \norm{h^0}_{\mathcal{E}^{p}_\nu} +\frac{1}{\eps}\norm{Q(h^0,h^0)+2Q(h^0,h^1)}_{\mathcal{E}^{p}},
\end{split}
\end{equation*}
where we used the scalar product notation to refer to the product operator appearing in $W^{\alpha,1}_vW^{\beta,p}_x$ when one differentiates $\norm{h}_{W^{\alpha,1}_vW^{\beta,p}_x\left(\langle v \rangle^k\right)}$ (of the same form as $\eqref{dissipativitystart}$). For the second inequality we used H\"older inequality between $L^p_x$ and $L^{p/(p-1)}_x$ inside the product operator:
$$\int_{\T^d}\mbox{sgn}(h^0)\abs{h^0}^{p-1} F(h^0)\:dx \leq \norm{h^0}^{p-1}_{L^p_x}\norm{F(h^0)}_{L^p_x}.$$

\bigskip
Then estimating $Q$ using Lemma $\ref{lem:controlbilinearlargeE}$ yields

\begin{equation}\label{h0}
\frac{d}{dt}\norm{h^0}_{\mathcal{E}^{p}} \leq -\frac{1}{\eps^2}\left[\frac{\lambda_0}{\nu_0} -2\eps C_{\beta,p}\left(\norm{h^0}_{\mathcal{E}^{p}} + \frac{2}{\nu_0} \norm{h^1}_{\mathcal{E}^{p}_\nu}\right)\right]\norm{h^0}_{\mathcal{E}^{p}_\nu},
\end{equation}
we recall $\nu_0 = \inf\limits_{v \in \R^d}(\nu(v)) >0$.

\bigskip
Therefore, if
$$\norm{h^1}_{\mathcal{E}^{p}_\nu} \leq \eps^{-1}\frac{(\lambda_0-\lambda_0')}{8C_{\beta,p}} \quad\mbox{and}\quad \norm{h_{t=0}}_{\mathcal{E}^{p}} \leq \eps^{-1}\frac{(\lambda_0-\lambda_0')}{4\nu_0 C_{\beta,p}},$$
then $\norm{h^0}_{\mathcal{E}^{p}}$ is always decreasing in time with
$$\frac{d}{dt}\norm{h^0}_{\mathcal{E}^{p}}\leq -\frac{\lambda_0'}{\eps^2\nu_0}\norm{h^0}_{\mathcal{E}^{p}_\nu},$$
which hence yields the expected exponential decay by Gr\"onwall Lemma.

\end{proof}
\bigskip


\subsection{Study of equations $\eqref{decompositionh1}$ in $E$}\label{subsec:h1}

In the space $E = H^\beta_{x,v}\left(\mu^{-1/2}\right)$, solutions to the perturbed Boltzmann equation enjoy an exponential decay. More precisely, \cite{Bri1} derived a precise Gr\"onwall that we will now use to obtain estimates on the solution $h^1$. We will use the following shorthand notation
$$E_\nu = H^\beta_{x,v}\left(\mu^{-1/2}\nu^{1/2}\right)$$
\par In this section we use the previous theorem to obtain exponential decay of $h^1$ in $E$. This result is stated in the following proposition, where $C^0_t$ denotes the space of time-continuous functions.

\bigskip
\begin{prop}\label{prop:h1}
Let $p=1$ or $p=2$, $0<\eps \leq \eps^d\leq 1$, $\beta \geq s_0$ and $\alpha \leq \beta$  ($\eps_d$ and $s_0$ being constructive constants that will be defined in Theorem $\ref{theo:adaptedBri1}$).
\\Let $h_{in}$ be in $\mathcal{E}^{p}$ and $h^0$ in $C^0_t\mathcal{E}^{p}$.
\\For any $\delta$ in $(0,\delta_{k,1}]$ and any $\lambda_0'$ in $(0,\lambda_0)$ ($\delta_{k,1}$ and $\lambda_0$ defined in Proposition $\ref{prop:hypodissipativityBeps}$) there exist $\eta_1,C_1 > 0$ such that
\\if
\begin{enumerate}
\item[(i)] $\norm{h_{in}}_{\mathcal{E}^{p}} \leq \eta_1,$
\item[(ii)] there exists $C_0 >0$ such that $\norm{h^0}_{\mathcal{E}^{p}} \leq C_0 e^{-\frac{\lambda_0 + \lambda_0'}{2\eps^2}t}\norm{h_{in}}_{\mathcal{E}^{p}},$
\item[(iii)] $h^1$ satisfies $h^1_{t=0} = 0$ and  is solution to 
$$\partial_t h^1 =  \mathcal{G}_\eps(h^{1}) + \frac{1}{\eps} Q(h^{1},h^{1}) + \mathcal{A}^{(\delta)}_\eps(h^{0})$$
\end{enumerate}
\noindent then
$$\norm{h^1}_{E} \leq C_1 e^{-\lambda_0't}\norm{h_{in}}_{\mathcal{E}^{p}}$$
The constants $C_1$ and $\eta_1$ are constructive and depends only on $\delta$, $\lambda_0'$ and the kernel of the Boltzmann operator.
\end{prop}
\bigskip

In order to prove Proposition $\ref{prop:h1}$ we need a new control on the bilinear term.
\par For any operator $\func{F}{E\times E}{E}$, we will say that $F$ satisfies the property $(H)$ if the following holds.
\paragraph{Property $(H)$:}
\begin{enumerate}
\item for all $g^1$, $g^2$ in $E$ we have $\pi_L\left(F(g^1,g^2)\right) = 0$, where $\pi_L$ is the orthogonal projection on $\mbox{Ker}\left(L\right)$ in $L^2_v\left(\mu^{-1/2}\right)$ (see $\eqref{piL}$),
\item for all $s'>0$ there exists $\func{\mathcal{F}_F^{s'}}{E\times E}{\R^+}$ such that for all multi-indexes $j$ and $l$ such that $|j|+|l| \leq s'$,
$$ \left|\langle\partial_l^jF(g^1,g^2),g^3\rangle_{L^2_{x,v}\left(\mu^{-1/2}\right)}\right| \leq  \mathcal{F}_F^{s'}(g^1,g^2)\norm{g^3}_{L^2_{x,v}\left(\mu^{-1/2}\nu^{1/2}\right)},$$
with $\mathcal{F}_F^{s'}\leq \mathcal{F}_F^{s'+1}$.
\end{enumerate}

\bigskip
\begin{lemma}\label{lem:controlbilinearsmallE}
The Boltzmann linear operator $Q$ satisfies the property $(H)$ with
$$ \forall s> d,\: \exists C_s>0,\quad \mathcal{F}^s_Q(g,h) \leq C_s \left[\norm{f}_{E}\norm{g}_{E_\nu} + \norm{f}_{E_\nu}\norm{g}_E\right].$$
\end{lemma}
\bigskip

The latter control on the bilinear part is from \cite{Bri1} Appendix $A.2$.

\bigskip
\begin{proof}[Proof of Proposition $\ref{prop:h1}$]
We state below the estimate derived in \cite{Bri1} (note that this is a version of \cite{Bri1} Theorem $2.4$ extended by estimates proved in \cite{Bri1} Propositions $2.2$ and $7.1$). 

\bigskip
\begin{theorem}\label{theo:adaptedBri1}
There exist $0<\eps_d \leq 1$ and $s_0$ in $\N$ such that for any $s\geq s_0$ and any $\lambda_0''$ in $(0,\lambda_0)$ there exists $\delta_s,\: C_s >0$ such that,
\begin{itemize}
\item for any $h_{in}$ in $H^s_{x,v}\left(\mu^{-1/2}\right)$ with 
$$\norm{h_{in}}_{H^s_{x,v}\left(\mu^{-1/2}\right)}\leq \delta_s,$$
\item for any operator $\func{F}{H^s_{x,v}\left(\mu^{-1/2}\right)\times H^s_{x,v}\left(\mu^{-1/2}\right)}{H^s_{x,v}\left(\mu^{-1/2}\right)}$ satisfying the property $\emph{(H)}$;
\end{itemize}
Then for all $0<\eps \leq \eps_d$ and for all $g^1$, $g^2$ in $H^s_{x,v}\left(\mu^{-1/2}\right)$, if $h$ is a solution to
$$\left\{\begin{array}{rl} &\displaystyle{\partial_t h = G_\eps(h) + \frac{1}{\eps}F(g^1,g^2)} \vspace{2mm} \\ \vspace{2mm} &\displaystyle{h_{t=0} = h_{in},} \end{array}\right.$$
and $h$ is in $\emph{\mbox{Ker}}\left(G_\eps\right)$ for all time, then
$$\forall t \in \R^+, \quad \frac{d}{dt}\norm{h}^2_{H^s_{x,v}\left(\mu^{-1/2}\right)} \leq - \frac{2\lambda_0''}{\nu_0^2}\norm{h}^2_{H^s_{x,v}\left(\mu^{-1/2}\nu\right)}+ C_s \left(\mathcal{F}_F^{s}(g^1,g^2)\right)^2.$$
\end{theorem}
\bigskip

Now, let $\lambda''$ be in $(0,\lambda_0)$, $s\geq s_0$ and $0<\eps\leq \eps_d$.
\\ The proof of Proposition $\ref{prop:h1}$ will be done in two steps. First we study the projection of $h^1$ onto $\mbox{Ker}\left(G_\eps\right)$ and then its orthogonal part.

\bigskip
\textbf{Estimate on the projection part.} We have that, see the decomposition $\eqref{decompositionh}$, that $h^1 = h - h^0$ with $h$ solution to the perturbed Boltzmann equation and thus satisfying $\Pi_{\mathcal{G}}(h)=0$. We therefore have that
$$\Pi_{\mathcal{G}}(h^1) = -\Pi_{\mathcal{G}}(h^0).$$
Moreover, Theorem $\ref{theo:linLqLp}$ tells us that $\Pi_{\mathcal{G}}$ and $\Pi_{G}$ coincide on $E$ and thus
$$\Pi_{G}(h^1) = -\Pi_{\mathcal{G}}(h^0),$$
and assumption $(ii)$ together with the shape of $\Pi_{\mathcal{G}}$ (see $\eqref{projectionLqLp}$), there exists a constant $C_\Pi >0$, depending only on the dimension $d$ and $s$ and the constant $C_0$, such that
\begin{equation}\label{projectionh1}
\norm{\Pi_{G}(h^1)}_{E_\nu} \leq C_\Pi e^{-\frac{\lambda_0+\lambda_0'}{2\eps^2}t}\norm{h_{in}}_{\mathcal{E}^{p}}.
\end{equation}

\bigskip
\textbf{Estimate on the orthogonal part.} Applying $\Pi^\bot_{G} = \mbox{Id} - \Pi_{G}$, the orthogonal projection onto $\left(\mbox{Ker}\left(G_\eps\right)\right)^\bot$ in $L^2_{x,v}\left(\mu^{-1/2}\right)$, to the differential equation satisfied by $h^1$ yields
\begin{eqnarray}
\partial_t \left(\Pi^\bot_{G}(h^1)\right) &=& G_\eps(h^1) + \Pi^\bot_{G}\left(\frac{1}{\eps}Q(h^1,h^1) + \mathcal{A}^{(\delta)}_\eps(h^{0})\right) \nonumber
\\&=& G_\eps\left(\Pi^\bot_{G}(h^1)\right) + \Pi^\bot_{G}\left(\frac{1}{\eps}Q(h^1,h^1) + \mathcal{A}^{(\delta)}_\eps(h^{0})\right). \label{orthogonalh1start}
\end{eqnarray}
Moreover, we have by definition of $\Pi_G$ and $\pi_L$ (see $\eqref{projectionLqLp}$ and $\eqref{piL}$) that 
$$\left(\pi_L (h) = 0 \right) \implies \left(\Pi_G (h) = 0\right)$$
and therefore
$$\Pi^\bot_{G}\left(Q(h^1,h^1)\right) = Q(h^1,h^1),$$
since $Q$ satisfies property $(H).1.$ by Lemma $\ref{lem:controlbilinearsmallE}$. Plugging the latter equality into $\eqref{orthogonalh1start}$ gives
$$\partial_t \left(\Pi^\bot_{G}(h^1)\right) = G_\eps\left(\Pi^\bot_{G}(h^1)\right) + \frac{1}{\eps}Q(h^1,h^1) + \Pi^\bot_{G}\left(\mathcal{A}^{(\delta)}_\eps(h^{0})\right).$$

\bigskip
By definition, $\Pi^\bot_{G}(h^1)$ is in $\left(\mbox{Ker}\left(G_\eps\right)\right)^\bot$ for all time and thanks to the control on the Boltzmann operator $Q$ in $E$ (Lemma $\ref{lem:controlbilinearsmallE}$), we are able to use Theorem $\ref{theo:adaptedBri1}$ with $\lambda_0 > \lambda_0'$ to which we have to add the source term $\Pi^\bot_{G}\left(\mathcal{A}^{(\delta)}_\eps(h^{0})\right)$. This yields the following differential inequality, where we denote by $C$ any positive constant independent of $\eps$,
\begin{eqnarray}
&&\frac{d}{dt}\norm{\Pi^\bot_{G}(h^1)}^2_{E} \label{orthogonalh1mid}
\\&&\quad\quad\quad\leq -\frac{2\lambda_0''}{\nu_0^2}\norm{\Pi^\bot_{G}(h^1)}_{E_\nu}^2 + C \left(\mathcal{F}_Q^{s}(h^1,h^1)\right)^2 + \abs{\langle \Pi^\bot_{G}\left(\mathcal{A}^{(\delta)}_\eps(h^{0})\right), \Pi^\bot_{G}(h^1) \rangle_E} \nonumber
\\&&\quad\quad\quad\leq -\frac{2\lambda_0''}{\nu_0^2}\norm{\Pi^\bot_{G}(h^1)}_{E_\nu}^2 + C \norm{h^1}^2_E\norm{h^1}^2_{E_\nu} + \norm{\Pi^\bot_{G}\left(\mathcal{A}^{(\delta)}_\eps(h^{0})\right)}_E\norm{\Pi^\bot_{G}(h^1)}_E, \nonumber
\end{eqnarray}
where we used a Cauchy-Schwarz inequality on the last term on the right-hand side.
\par Then we can decompose $h^1 = \Pi_{G}(h^1) + \Pi^\bot_{G}(h^1)$ to get first
\begin{equation*}
\begin{split}
\norm{h^1}^2_E\norm{h^1}^2_{E_\nu} \leq &4 \norm{\Pi^\bot_{G}(h^1)}^2_E\norm{\Pi^\bot_{G}(h^1)}^2_{E_\nu} + \frac{8}{\nu_0^2}\norm{\Pi_{G}(h^1)}^2_{E_\nu}\norm{\Pi^\bot_{G}(h^1)}^2_{E_\nu} 
\\&+\frac{4}{\nu_0^2}\norm{\Pi_{G}(h^1)}^4_{E_\nu},
\end{split}
\end{equation*}
into which we can plug the control on $\norm{\Pi_{G}(h^1)}^2_{E_\nu}$ we derived in $\eqref{projectionh1}$ to obtain, with $\norm{h_{in}}\leq \eta_1$,
\begin{equation}\label{h1h1}
\norm{h^1}^2_E\norm{h^1}^2_{E_\nu} \leq 4 \norm{\Pi^\bot_{G}(h^1)}^2_E\norm{\Pi^\bot_{G}(h^1)}^2_{E_\nu} + C\eta_1^2\norm{\Pi^\bot_{G}(h^1)}^2_{E_\nu}  +C e^{-\frac{2(\lambda_0 +\lambda_0')}{\eps^2}t}\norm{h_{in}}_{\mathcal{E}^{p}}^4.
\end{equation}
\par And finally, this inequality together with assumption $(ii)$ gives the existence of a constant $C_A >0$ such that
\begin{equation}\label{h0h1}
\norm{\Pi^\bot_{G}\left(\mathcal{A}^{(\delta)}_\eps(h^{0})\right)}_E\norm{\Pi^\bot_{G}(h^1)}_E \leq \frac{C_A}{\eps^2}\norm{h_{in}}_{\mathcal{E}^{p}} e^{-\frac{\lambda_0+\lambda_0'}{2\eps^2}t}\norm{\Pi^\bot_{G}(h^1)}_E.
\end{equation}

\bigskip
We plug $\eqref{h1h1}$ and $\eqref{h0h1}$ into $\eqref{orthogonalh1mid}$ and obtain, with $C$ and $C'$ being positive constants independent of $\eps$,
\begin{equation*}
\begin{split}
\frac{d}{dt}\norm{\Pi^\bot_{G}(h^1)}^2_{E} \leq &-\left[\frac{2\lambda_0''}{\nu_0^2} -\left(4\norm{\Pi^\bot_{G}(h^1)}^2_{E} + C\eta_1^2\right) \right]\norm{\Pi^\bot_{G}(h^1)}^2_{E_\nu} 
\\&+ C'\left(\norm{h_{in}}_{\mathcal{E}^{p}}^4 + \frac{1}{\eps^2}\norm{h_{in}}_{\mathcal{E}^{p}}\norm{\Pi^\bot_{G}(h^1)}_E\right)e^{-\frac{\lambda_0+\lambda_0'}{2\eps^2}t}.
\end{split}
\end{equation*}

We now choose $\eta_1$ sufficiently small so that 
$$C\eta_1^2 \leq \frac{\lambda_0'' - \lambda_0'}{\nu_0^2},$$
which in turns implies
\begin{equation}\label{orthogonalh1final}
\begin{split}
\frac{d}{dt}\norm{\Pi^\bot_{G}(h^1)}^2_{E} \leq &-\left[\frac{\lambda_0''+\lambda_0'}{\nu_0^2} -4\norm{\Pi^\bot_{G}(h^1)}^2_{E}  \right]\norm{\Pi^\bot_{G}(h^1)}^2_{E_\nu} 
\\&+C'\left(\norm{h_{in}}_{\mathcal{E}^{p}}^4 + \frac{1}{\eps^2}\norm{h_{in}}_{\mathcal{E}^{p}}\norm{\Pi^\bot_{G}(h^1)}_E\right)e^{-\frac{\lambda_0+\lambda_0'}{2\eps^2}t}.
\end{split}
\end{equation}

\bigskip
We define
$$\eta_* = \frac{\lambda_0''-\lambda_0'}{4\nu_0^2}. $$
 We have that $h^1_{t=0} = 0$ so we can define 
$$t_0 = \sup\{t>0,\quad \norm{\Pi^\bot_{G}(h^1)}^2_{E} < \eta_*\}.$$
Suppose that $t_0 < +\infty$,  we therefore have for all $t$ in $[0,t_0]$

$$\frac{d}{dt}\norm{\Pi^\bot_{G}(h^1)}^2_{E}\leq -\frac{2\lambda_0'}{\nu_0^2}\norm{\Pi^\bot_{G}(h^1)}^2_{E_\nu} + C' \left(\norm{h_{in}}_{\mathcal{E}^{p}}^4 + \frac{\sqrt{\eta_*}}{\eps^2}\norm{h_{in}}_{\mathcal{E}^{p}}\right)e^{-\frac{\lambda_0+\lambda_0'}{2\eps^2}t},$$
which gives
$$\forall t \in [0,t_0], \quad \frac{d}{dt}\norm{\Pi^\bot_{G}(h^1)}^2_{E}\leq -2\lambda'_0 \norm{\Pi^\bot_{G}(h^1)}^2_{E} + C'\left(\norm{h_{in}}_{\mathcal{E}^{p}}^4 + \frac{\sqrt{\eta_*}}{\eps^2}\norm{h_{in}}_{\mathcal{E}^{p}}\right)e^{-\frac{\lambda_0+\lambda_0'}{2\eps^2}t},$$
and by Gronwall lemma with $\Pi^\bot_{G}(h^1)_{(t=0)}=0$,
\begin{eqnarray*}
\forall t \in [0,t_0],\quad \norm{\Pi^\bot_{G}(h^1)}^2_{E} &\leq& C'\left(\norm{h_{in}}_{\mathcal{E}^{p}}^4 + \frac{\sqrt{\eta_*}}{\eps^2}\norm{h_{in}}_{\mathcal{E}^{p}}\right)\left(\int_0^t e^{-\frac{\lambda_0+\lambda_0'}{2\eps^2}s}\:e^{2\lambda_0's}\:ds\right)e^{-2\lambda_0't}
\\ &\leq & C'\left(\eps^{2}\norm{h_{in}}_{\mathcal{E}^{p}}^4 + \sqrt{\eta_*}\norm{h_{in}}_{\mathcal{E}^{p}}\right)\left(\int_0^{+\infty} e^{-\frac{\lambda_0-\lambda_0'}{2}u}\:du\right)e^{-2\lambda_0't},
\end{eqnarray*}
where we used the change of variable $u = \eps^{-2}s$ and we considered $\eps \leq 1/4$ (which only amounts to decreasing $\eps_d$).

\bigskip
Hence, there exists $K >0$ independent of $\eps$ such that
$$\forall t \in [0,t_0], \quad \norm{\Pi^\bot_{G}(h^1)}^2_{E}\leq K(\eta_1^4+\eta_1\sqrt{\eta_*}).$$

If we thus chose $\eta_1$ sufficiently small such that $(\eta_1^4+\eta_1\sqrt{\eta_*})K <\eta_*/2$ we reach a contradiction when $t$ goes to $t_0$ since $\norm{\Pi^\bot_{G}(h^1)}^2_{E}(t_0) \geq \eta_*$. Therefore, choosing $\eta_1$ small enough independently on $\eps$ implies first that $t_0 = +\infty$ and second that

\begin{equation}\label{orthogonalh1}
\forall t \in [0,+\infty),\quad \norm{\Pi^\bot_{G}(h^1)}^2_{E} \leq  C \norm{h_{in}}_{\mathcal{E}^{p}}^2e^{-2\lambda_0't}.
\end{equation}

\paragraph{End of the proof.} By just decomposing $h^1$ into its projection and orthogonal part and using the estimates $\eqref{projectionh1}$ and $\eqref{orthogonalh1}$ gives the expected exponential decay for $h^1$ in $E$.
\end{proof}
\bigskip


\subsection{Proof of Theorem $\ref{theo:aprioridecomposition}$}\label{subsec:proofapriori}

Let $p=1$ or $p=2$, $\lambda''$ be in $(0,\lambda_0)$, $\beta\geq \beta_0 = s_0$ and $0<\eps\leq \eps_d$. All the constants used in this section are the ones constructed in Proposition $\ref{prop:h0}$ with $(\lambda_0+\lambda_0')/2$ and Proposition $\ref{prop:h1}$ with $\lambda_0'$.
\par $E$ is continuously embedded in $\mathcal{E}^{p}_\nu$ because $L^2_v\left(\mu^{-1/2}\right) \subset L^2_v\left(\langle v \rangle^k\right)$ (mere Cauchy-Schwarz inequality) and $L^2_x \subset L^1_x$ because $\T^d$ is bounded. Hence, there exists $C_{E,\mathcal{E}}>0$ such that
\begin{equation}\label{equivalenceEpE}
\frac{1}{\nu_0}\norm{\cdot}_{\mathcal{E}^{p}} \leq \norm{\cdot}_{\mathcal{E}^{p}_\nu} \leq C_{E,\mathcal{E}}\norm{\cdot}_{E}.
\end{equation}

\bigskip
We define
$$\eta = \min\left(\eta_0,\eta_1, \frac{\eta_0}{2C_{E,\mathcal{E}}C_1}\right),$$
and we assume $\norm{h_{in}}_{\mathcal{E}^{p}} \leq \eta$. Since $h^1_{t=0} = 0$ we also define 
$$t_0 = \sup\{t>0,\quad \norm{h^1}_{\mathcal{E}_\nu^p} < \eta_0\}.$$

\bigskip
Suppose that $t_0 < +\infty$. Then, thanks to Proposition $\ref{prop:h0}$ we have that
$$\forall t \in [0,t_0], \quad \norm{h^0}_{\mathcal{E}^{p}} \leq \norm{h_{in}}_{\mathcal{E}^{p}} e^{-\frac{\lambda_0+\lambda_0'}{2\eps^2}t}.$$
We can thus apply Proposition $\ref{prop:h1}$ and get
$$\forall t \in [0,t_0], \quad \norm{h^1}_{E} \leq C_1\norm{h_{in}}_{\mathcal{E}^{p}} e^{-\lambda_0't}\leq C_1\eta \leq \frac{\eta_0}{2C_{E,\mathcal{E}}},$$
which is in contradiction with the definition of $t_0$ thanks to $\eqref{equivalenceEpE}$. Therefore $t_0=+\infty$ and we have the expected exponential decay stated in Theorem $\ref{theo:aprioridecomposition}$ for all time.




\bibliographystyle{acm}
\bibliography{bibliography}

\begin{thebibliography}{10}

\bibitem{Ark1}
{\sc Arkeryd, L.}
\newblock Stability in {$L^1$} for the spatially homogeneous {B}oltzmann
  equation.
\newblock {\em Arch. Rational Mech. Anal. 103}, 2 (1988), 151--167.

\bibitem{BM}
{\sc Baranger, C., and Mouhot, C.}
\newblock Explicit spectral gap estimates for the linearized {B}oltzmann and
  {L}andau operators with hard potentials.
\newblock {\em Rev. Mat. Iberoamericana 21}, 3 (2005), 819--841.

\bibitem{BGL}
{\sc Bardos, C., Golse, F., and Levermore, D.}
\newblock Fluid dynamic limits of kinetic equations. {I}. {F}ormal derivations.
\newblock {\em J. Statist. Phys. 63}, 1-2 (1991), 323--344.

\bibitem{BU}
{\sc Bardos, C., and Ukai, S.}
\newblock The classical incompressible {N}avier-{S}tokes limit of the
  {B}oltzmann equation.
\newblock {\em Math. Models Methods Appl. Sci. 1}, 2 (1991), 235--257.

\bibitem{BerLof}
{\sc Bergh, J., and L{\"o}fstr{\"o}m, J.}
\newblock {\em Interpolation spaces. {A}n introduction}.
\newblock Springer-Verlag, Berlin-New York, 1976.
\newblock Grundlehren der Mathematischen Wissenschaften, No. 223.

\bibitem{Bob1}
{\sc Bobyl{\"e}v, A.~V.}
\newblock The method of the {F}ourier transform in the theory of the
  {B}oltzmann equation for {M}axwell molecules.
\newblock {\em Dokl. Akad. Nauk SSSR 225}, 6 (1975), 1041--1044.

\bibitem{Bob2}
{\sc Bobyl{\"e}v, A.~V.}
\newblock The theory of the nonlinear spatially uniform {B}oltzmann equation
  for {M}axwell molecules.
\newblock In {\em Mathematical physics reviews, {V}ol.\ 7}, vol.~7 of {\em
  Soviet Sci. Rev. Sect. C Math. Phys. Rev.} Harwood Academic Publ., Chur,
  1988, pp.~111--233.

\bibitem{BouchDes}
{\sc Bouchut, F., and Desvillettes, L.}
\newblock Averaging lemmas without time {F}ourier transform and application to
  discretized kinetic equations.
\newblock {\em Proc. Roy. Soc. Edinburgh Sect. A 129}, 1 (1999), 19--36.

\bibitem{BoudDes}
{\sc Boudin, L., and Desvillettes, L.}
\newblock On the singularities of the global small solutions of the full
  {B}oltzmann equation.
\newblock {\em Monatsh. Math. 131}, 2 (2000), 91--108.

\bibitem{Br}
{\sc Brezis, H.}
\newblock {\em Analyse fonctionnelle}.
\newblock Collection Math\'ematiques Appliqu\'ees pour la Ma\^\i trise.
  [Collection of Applied Mathematics for the Master's Degree]. Masson, Paris,
  1983.
\newblock Th{\'e}orie et applications. [Theory and applications].

\bibitem{Bri1}
{\sc Briant, M.}
\newblock A constructive method from {B}oltzmann to incompressible
  {N}avier-{S}tokes on the torus.

\bibitem{Ca2}
{\sc Carleman, T.}
\newblock {\em Probl\`emes math\'ematiques dans la th\'eorie cin\'etique des
  gaz}.
\newblock Publ. Sci. Inst. Mittag-Leffler. 2. Almqvist \& Wiksells Boktryckeri
  Ab, Uppsala, 1957.

\bibitem{Ce}
{\sc Cercignani, C.}
\newblock {\em The {B}oltzmann equation and its applications}, vol.~67 of {\em
  Applied Mathematical Sciences}.
\newblock Springer-Verlag, New York, 1988.

\bibitem{Ce1}
{\sc Cercignani, C., Illner, R., and Pulvirenti, M.}
\newblock {\em The mathematical theory of dilute gases}, vol.~106 of {\em
  Applied Mathematical Sciences}.
\newblock Springer-Verlag, New York, 1994.

\bibitem{GST}
{\sc Gallagher, I., Saint-Raymond, L., and Texier, B.}
\newblock From newton to boltzmann: the case of short-range potentials.

\bibitem{Go}
{\sc Golse, F.}
\newblock From kinetic to macroscopic models, 1998.

\bibitem{Gr1}
{\sc Grad, H.}
\newblock Principles of the kinetic theory of gases.
\newblock In {\em Handbuch der {P}hysik (herausgegeben von {S}. {F}l\"ugge),
  {B}d. 12, {T}hermodynamik der {G}ase}. Springer-Verlag, Berlin, 1958,
  pp.~205--294.

\bibitem{Gr2}
{\sc Grad, H.}
\newblock Asymptotic theory of the {B}oltzmann equation. {II}.
\newblock In {\em Rarefied {G}as {D}ynamics ({P}roc. 3rd {I}nternat. {S}ympos.,
  {P}alais de l'{UNESCO}, {P}aris, 1962), {V}ol. {I}}. Academic Press, New
  York, 1963, pp.~26--59.

\bibitem{Gr}
{\sc Grad, H.}
\newblock Asymptotic equivalence of the {N}avier-{S}tokes and nonlinear
  {B}oltzmann equations.
\newblock In {\em Proc. {S}ympos. {A}ppl. {M}ath., {V}ol. {XVII}}. Amer. Math.
  Soc., Providence, R.I., 1965, pp.~154--183.

\bibitem{GMM}
{\sc Gualdani, M.~P., Mischler, S., and Mouhot, C.}
\newblock Factorization for non-symmetric operators and exponential
  {H}-theorem.

\bibitem{Gu1}
{\sc Guo, Y.}
\newblock Classical solutions to the {B}oltzmann equation for molecules with an
  angular cutoff.
\newblock {\em Arch. Ration. Mech. Anal. 169}, 4 (2003), 305--353.

\bibitem{Gu4}
{\sc Guo, Y.}
\newblock Boltzmann diffusive limit beyond the {N}avier-{S}tokes approximation.
\newblock {\em Comm. Pure Appl. Math. 59}, 5 (2006), 626--687.

\bibitem{Ka}
{\sc Kato, T.}
\newblock {\em Perturbation theory for linear operators}.
\newblock Classics in Mathematics. Springer-Verlag, Berlin, 1995.
\newblock Reprint of the 1980 edition.

\bibitem{La}
{\sc Lanford, III, O.~E.}
\newblock Time evolution of large classical systems.
\newblock In {\em Dynamical systems, theory and applications ({R}econtres,
  {B}attelle {R}es. {I}nst., {S}eattle, {W}ash., 1974)}. Springer, Berlin,
  1975, pp.~1--111. Lecture Notes in Phys., Vol. 38.

\bibitem{Le}
{\sc Leray, J.}
\newblock Sur le mouvement d'un liquide visqueux emplissant l'espace.
\newblock {\em Acta Math. 63}, 1 (1934), 193--248.

\bibitem{Mo1}
{\sc Mouhot, C.}
\newblock Explicit coercivity estimates for the linearized {B}oltzmann and
  {L}andau operators.
\newblock {\em Comm. Partial Differential Equations 31}, 7-9 (2006),
  1321--1348.

\bibitem{Mo3}
{\sc Mouhot, C.}
\newblock Rate of convergence to equilibrium for the spatially homogeneous
  {B}oltzmann equation with hard potentials.
\newblock {\em Communications in Mathematical Physics 261\/} (2006), 629--672.

\bibitem{MN}
{\sc Mouhot, C., and Neumann, L.}
\newblock Quantitative perturbative study of convergence to equilibrium for
  collisional kinetic models in the torus.
\newblock {\em Nonlinearity 19}, 4 (2006), 969--998.

\bibitem{MouVi}
{\sc Mouhot, C., and Villani, C.}
\newblock Regularity theory for the spatially homogeneous {B}oltzmann equation
  with cut-off.
\newblock {\em Arch. Ration. Mech. Anal. 173}, 2 (2004), 169--212.

\bibitem{Sa}
{\sc Saint-Raymond, L.}
\newblock {\em Hydrodynamic limits of the {B}oltzmann equation}, vol.~1971 of
  {\em Lecture Notes in Mathematics}.
\newblock Springer-Verlag, Berlin, 2009.

\bibitem{Uk}
{\sc Ukai, S.}
\newblock On the existence of global solutions of mixed problem for non-linear
  {B}oltzmann equation.
\newblock {\em Proc. Japan Acad. 50\/} (1974), 179--184.

\bibitem{UkYa}
{\sc Ukai, S., and Yang, T.}
\newblock Mathematical theory of the {B}oltzmann equation.
\newblock Lecture Notes Series, no. 8, Liu Bie Ju Centre for Mathematical
  Sciences, City University of Hong Kong, 2006.

\bibitem{Vi}
{\sc Villani, C.}
\newblock Limites hydrodynamiques de l'\'equation de {B}oltzmann (d'apr\`es
  {C}. {B}ardos, {F}. {G}olse, {C}. {D}.\ {L}evermore, {P}.-{L}.\ {L}ions,
  {N}.\ {M}asmoudi, {L}. {S}aint-{R}aymond).
\newblock {\em Ast\'erisque}, 282 (2002), Exp. No. 893, ix, 365--405.
\newblock S{\'e}minaire Bourbaki, Vol. 2000/2001.

\bibitem{Vi2}
{\sc Villani, C.}
\newblock A review of mathematical topics in collisional kinetic theory.
\newblock In {\em Handbook of mathematical fluid dynamics, {V}ol. {I}}.
  North-Holland, Amsterdam, 2002, pp.~71--305.

\bibitem{Yu}
{\sc Yu, H.}
\newblock Global classical solutions of the {B}oltzmann equation near
  {M}axwellians.
\newblock {\em Acta Math. Sci. Ser. B Engl. Ed. 26}, 3 (2006), 491--501.

\end{thebibliography}


\signmb\signsm
\signcm
\end{document}